\documentclass[journal,comsoc]{IEEEtran}
\usepackage{amsmath,amsfonts}
\usepackage{algorithm}
\usepackage{array}
\usepackage{algpseudocode}
\usepackage[caption=false,font=normalsize,labelfont=sf,textfont=sf]{subfig}
\usepackage{textcomp}
\usepackage{stfloats}
\usepackage{url}
\usepackage{verbatim}
\usepackage{amsthm,amssymb}
\usepackage{graphicx}
\usepackage{color}
\usepackage{cite}
\usepackage{multirow}
\newtheorem{theorem}{Theorem}
\begin{document}

% \title{The Rewiring Strategy Effectively Changes The Network Degree Correlation}
% \title{Attack Degree Correlation Coefficient Under Rewiring Strategy}
\title{Improving Network Degree Correlation by Degree-preserving Rewiring}
\author{Shuo Zou, Bo Zhou, and~Qi~Xuan,~\IEEEmembership{Senior Member,~IEEE,}
\thanks{This work was supported in part by the National Natural Science Foundation of China under Grant 61973273, by the Zhejiang Provincial Natural Science Foundation of China under Grant LR19F030001, by the National Key R\&D Program of China under Grant 2020YFB1006104, and by the Research and Development Center of Transport Industry of New Generation of Artificial Intelligence Technology. \emph{(Corresponding authors: Qi Xuan.)}}
\thanks{
All authors are with the Institute of Cyberspace Security, College of Information Engineering, Zhejiang University of Technology, Hangzhou 310023, China.} 
%B. Zhou is also with the Department of Intelligent Control, Zhejiang Institute of Communications, Hangzhou 311112, China.}
}
% \thanks{This work was supported in part by the Key R\&D Program of Zhejiang under Grant 2022C01018, by the National Natural Science Foundation of China under Grants U21B2001 and 61973273, by the Zhejiang Provincial Natural Science Foundation of China under Grant LR19F030001, and by the Research and Development Center of Transport Industry of New Generation of Artificial Intelligence Technology under Grant 202102H. \emph{(Corresponding author: Qi Xuan.)}}
% \author{Shuo Zou, 
%         Bo , Yongchao Mao, Jinhuan Wang, Shanqing Yu, and~Qi~Xuan,~\IEEEmembership{Member,~IEEE}
%~\IEEEmembership{Staff,~IEEE,}
        % <-this % stops a space



\maketitle

\begin{abstract}
Degree correlation is a crucial measure in networks, significantly impacting network topology and dynamical behavior. The degree sequence of a network is a significant characteristic, and altering network degree correlation through degree-preserving rewiring poses an interesting problem. In this paper, we define the problem of maximizing network degree correlation through a finite number of rewirings and use the assortativity coefficient to measure it. We analyze the changes in assortativity coefficient under degree-preserving rewiring and establish its relationship with the $s-$metric. Under our assumptions, we prove the problem to be monotonic and submodular, leading to the proposal of the GA method to enhance network degree correlation. By formulating an integer programming model, we demonstrate that the GA method can effectively approximate the optimal solution and validate its superiority over other baseline methods through experiments on three types of real-world networks. Additionally, we introduce three heuristic rewiring strategies, EDA, TA and PEA, and demonstrate their applicability to different types of networks. Furthermore, we extend the application of our proposed rewiring strategies to investigate their impact on several spectral robustness metrics based on the adjacency matrix, revealing that GA effectively improves network robustness, while TA performs well in enhancing the robustness of power networks, PEA exhibits promising performance in routing networks, and both heuristic methods outperform other baseline methods in flight networks. Finally, we explored the robustness of several centrality metrics in the network while enhancing network degree correlation using the GA method. We found that, for disassortative real networks, closeness centrality and eigenvector centrality are typically robust. When focusing on the top-ranked nodes, we observed that all centrality metrics remain robust in disassortative networks.
% Degree correlation is an important characteristic of networks, which is usually quantified by the assortativity coefficient. However, concerns arise about changing the assortativity coefficient of a network when networks suffer from adversarial attacks. In this paper, we analyze the factors that affect the assortativity coefficient and study the optimization problem of maximizing or minimizing the assortativity coefficient ($r$) in rewired networks with $k$ pairs of edges. We propose a greedy algorithm and formulate the optimization problem using integer programming to obtain the optimal solution for this problem. Through experiments, we demonstrate the reasonableness and effectiveness of our proposed algorithm. For example, rewired edges 10\% in the ER network, the assortativity coefficient improved by 60\%.

\end{abstract}


\section{Introduction}
\IEEEPARstart{C}{omplex} networks serve as powerful tools for abstractly representing real-world systems, where individual units are represented as nodes, and interactions between these units are represented as edges.
%\cite{newman2003structure,estrada2012structure}. 
Therefore, research on complex networks has experienced tremendous growth in recent years. Various network properties, including the degree sequence\cite{chung2002connected,chatterjee2011random}, degree correlation\cite{park2003origin,mahadevan2006systematic} and clustering coefficient\cite{saramaki2007generalizations,mcassey2015clustering} are extensively utilized in complex network analysis to assess the topological structure of networks. 
%The degree correlation, one of the most intriguing properties in complex networks, describes the relationship between the degrees of connected nodes. This crucial topological characteristic plays an important role in network stability\cite{friedel2007influence}, attack robustness\cite{huang2011robustness}, network controllability\cite{zhang2019evolution}, information propagation\cite{vega2020influence} and other related phenomena\cite{noldus2013effect,zhou2014memetic,oh2018complex,osat2017optimal,olvera2021pagerank,li2021percolation}. 

In the field of complex networks, systems represented as networks often have different properties in reality. One of the most interesting properties is degree correlation. It represents the relationship between the degrees of connected nodes, such as whether nodes with large degrees tend to be connected to other nodes with large degrees or to nodes with small degrees. Degree correlation is an important concept in network analysis. For example, degree correlation in social networks may reflect the idea that popular individuals tend to know other popular individuals. Similarly, in citation networks, papers that are highly cited may tend to cite other highly cited papers. A network is referred to as assortative when high-degree nodes tend to connect to other high-degree nodes, and low-degree nodes tend to connect to other low-degree nodes. On the other hand, a network is called disassortative when high-degree nodes tend to connect to low-degree nodes, and low-degree nodes tend to connect to high-degree nodes. A network is considered neutral when there is no preferential tendency in connections between nodes.

There are several measures of degree correlation for undirected networks. The most popular among them is the assortativity coefficient, denoted as $r$. It is the Pearson correlation coefficient between the degrees of connected nodes in the network. The assortativity coefficient is a normalized measure, ranging between -1 and 1. It was initially introduced by Newman\cite{newman2002assortative,newman2003mixing}. Li~\emph{et al.}\cite{li2005towards} proposed the $s$-metric, which is obtained by calculating the product of the degrees of connected nodes. When using this measure, normalization is often required. This involves computing the maximum and minimum $s$-metric under the current degree sequence, which can be challenging. When the degree sequence of the network remains unchanged, the definition of the assortativity coefficient includes the $s$-metric. Therefore, this paper primarily uses the assortativity coefficient to measure the degree correlation in networks.

The problem considered in this paper is as follows: Given a simple undirected network and a budget, we aim to maximally improve the degree correlation of the network while meeting the budget constraint through the modification of its topological structure. The changes to the network's topological structure can take various forms, including edge addition, edge deletion, and edge rewiring. We primarily consider edge rewiring, altering the network's topological structure without changing the node degrees. This is practically meaningful since, in real-world networks, nodes often have capacity constraints. For instance, increasing the number of flights between airports may raise operational costs, which could be impractical in the short term. However, adjusting flights between airports through rewiring is a relatively straightforward approach. In router networks, rewiring connections between routers allows adjustments without altering their loads.

There is some research on changing network degree correlation through rewiring. Xulvi~\emph{et al.}\cite{xulvi2005changing} proposed two algorithms that aim to achieve the desired degree correlation in a network by producing assortative and disassortative mixing, respectively.  Li~\emph{et al.}\cite{jing2016algorithm} developed a probabilistic attack method that increases the chances of rewiring the edges between nodes of higher degrees, leading to a network with a higher degree of assortativity. Geng~\emph{et al.}\cite{geng2021global} introduced a global disassortative rewiring strategy aimed at establishing connections between high-degree nodes and low-degree nodes through rewiring, resulting in a higher level of disassortativity within the network. However, the mentioned works did not consider the rewiring budget. This paper primarily investigates how to maximize the degree correlation of a network through rewiring under a limited budget.

Degree correlation is a crucial property in complex networks, and different types of networks exhibit varying degrees of degree correlation. These differences in degree correlation result in distinct topological characteristics\cite{jing2007effects, zhou2008generating,noldus2013effect}, such as the distribution of path lengths and Rich-club coefficient, within networks. The diverse effects of degree correlation play a significant role in processes like disease propagation\cite{chang2020impact,bogua2003epidemic} and also impact the robustness of networks\cite{zhou2014memetic,menche2010asymptotic,geng2021global}. In this paper, we mainly focus on examining the impact of our method on several robustness measures based on the network adjacency-spectrum, while altering degree correlation. This helps determine whether our method contributes to enhancing the robustness of the network. 

The robustness of centrality metrics in networks is also an important research question. It investigates whether centrality metrics can maintain robustness when the network's topology changes. Some researchers have studied the variations of various centrality metrics in networks when nodes or edges fail\cite{platig2013robustness,martin2019influence,niu2015robustness}. In this paper, we explore which centrality metrics in the network can maintain robustness while our rewiring methods improve network degree correlation.

In this paper, we investigated the problem of maximizing network degree correlation through a finite number of rewirings. Our contributions are summarized as follows:
 \begin{itemize}
\item We defined the problem of maximizing degree correlation and proposed the GA, EDA, TA, and PEA algorithms.
\item We proved that under our assumptions, the objective function is monotonic and submodular.
\item We validated that GA can effectively approximate the optimal solution and significantly improve network degree correlation on several real networks. Meanwhile, EDA,TA and PEA also demonstrated their respective advantages.
\item We applied these rewiring strategies to enhance network robustness and found that GA can effectively improve network robustness. Additionally, EDA, TA and PEA showed applicability to different types of networks for enhancing network robustness.
\item We analyzed the robustness of several centrality metrics when networks were rewired using the GA method. Our findings indicate that in disassortative real networks, closeness centrality and eigenvector centrality exhibit robustness. Furthermore, upon focusing on the top-ranked nodes, we observed that all centrality metrics maintain their robustness in disassortative networks.
 \end{itemize}
% In this paper, we investigate a greedy rewiring strategy designed to manipulate the assortativity coefficient of a network while preserving its degree sequence. We make the following contributions:
% \begin{itemize}
%     \item We define the assortativity coefficient attack and propose a rewiring strategy based on the original graph.
%     \item We have shown that the objective function under the rewiring strategy is monotone and submodular.
%     \item We propose a greedy rewiring strategy and demonstrate its effectiveness on 2 small real-world networks. Furthermore, we comprehensively evaluate our method by comparing it with $3$ baselines on $6$ datasets.
    
% \end{itemize}

The structure of the paper is as follows. In Sec. \ref{sec}, we introduce the degree correlation measure of networks, specifically the assortativity coefficient, and analyze its variation under degree-preserving rewiring. We also establish a connection between the assortativity coefficient and another degree correlation metric, the $s-$metric. In Sec. \ref{sec}, we define the problem of maximizing degree correlation through rewiring and analyze the objective function is monotonic and submodular, leading to the proposal of the GA strategy, and we describe three heuristic rewiring methods, EDA, TA and PEA. In Sec. \ref{thre}, we validate the rationality of our assumption and demonstrate that the GA method effectively approximates the optimal solution. Through experiments on different types of real networks, we demonstrate that GA can effectively enhance network degree correlation, while EDA, TA, and PEA are applicable to different network types. Additionally, we investigate the impact of these rewiring methods on the spectral robustness of networks, and explore the robustness of several centrality metrics in the network while enhancing network degree correlation using the GA method. Finally, Sec. \ref{thi} concludes with a summary of findings and outlines avenues for future research.

% The remainder of this paper is organized as follows. Section~\ref{sec} gives the definition of the problem, and Section~\ref{thre} describes the experimental results. Finally, we give the conclusion in Section~\ref{thi}.

\section{Methodology}\label{sec}
\subsection{Preliminaries and Ideas}
We consider an undirected and unweighted network $G=(V,E)$, where the set of vertex $V$ is a set of $N$ nodes, and $E$ is a set of edges $M$. The assortativity coefficient is a widely used measure to quantify the degree correlation in a network. In this paper, we primarily utilize the assortativity coefficient to measure the degree correlation of the network. 
The assortativity coefficient is defined as\cite{newman2003mixing}:
\begin{equation}
    \mathbf{r} = \frac{M^{-1}\sum_i^M(j_{i}k_{i})-[M^{-1}\sum_i^M\frac{1}{2}(j_{i}+k_{i})]^2}{M^{-1}\sum_i^M\frac{1}{2}(j_{i}^2+k_{i}^2)-[M^{-1}\sum_i^M\frac{1}{2}(j_{i}+k_{i})]^2}.
    \label{eq:1}
\end{equation}
where $k_i$ and $j_i$ are the degrees of the endpoins of the $i$th edge, respectively. 

The degree distribution is a crucial characteristic of a network as it reveals the connectivity patterns and the overall topology of the network. Therefore, we employ a rewiring strategy to alter the network's topology without changing the degree of each node in the network. The rewiring strategy is shown in Figure~\ref{fig:99}. We choose an edge pair $\langle (i,j),(k,l) \rangle$ from the original network $G$ that satisfies $(i,j) \in E$ and $(k,l) \in E$, which can be rewired as $(i,k)$ and $(j,l)$ if $(i,k), (j,l) \notin E$, or can be rewired as $(i,l)$ and $ (k,j)$ if $(i,l),(k,j) \notin E$. Obviously, the rewiring strategy does not change the degree of the nodes. According to Formula~\ref{eq:1}, $\sum_i^{M}\frac{1}{2}(j_i^2+k_i^2)$ and $\sum_i^{M}\frac{1}{2}(j_i+k_i)$ are also unchanged under the rewiring strategy. The rewiring strategy only affects the following formula:
\begin{equation}
    \mathbf{s} = \sum_i^M(j_ik_i).
    \label{eq:02}
\end{equation}
We can observe that $s$ is the $s$-metric proposed by Li~\emph{et al.}\cite{li2005towards} Typically, the $s$-metric needs to be normalized to quantify the degree correlation of the network. The normalized $s$-metric is defined by \cite{li2005towards,li2004first}:
\begin{equation}
    s^n = \frac{s-s_{min}}{s_{max}-s_{min}}.
    \label{eq:03}
\end{equation}
Here, $s_{min}$ and $s_{max}$ are the minimum and the maximum values of $s$ from networks with the same degree sequence. Typically, calculating $s_{min}$ and $s_{max}$ is not straightforward, so more often, the assortativity coefficient is used to measure the degree correlation of networks. However, under the rewiring strategy, the change in assortativity coefficient translates to the change in the $s$-metric, and their meanings are equivalent. Nevertheless, to distinctly represent the degree correlation of the network, we will still use the assortativity coefficient in the following paper.

When the edge pair $\langle(i,j),(k,l) \rangle$ is rewired to $\langle(i,k),(j,l)\rangle$, the change in the assortativity coefficient can be converted to the change in $s$, calculated as:
\begin{equation}
    value_{\langle(i,j),(k,l)\rangle}=(d_i d_k + d_j d_l)-(d_i d_j + d_k d_l). 
    \label{eq:04}
\end{equation}
where $d_i$ represents the degree of node $i$.  It is important to note that $value_{\langle(i,j),(k,l)\rangle}$ represents the rewiring of edge pair $\langle (i,j),(k,l) \rangle$ to $\langle (i,k),(j,l) \rangle$. The edge pair $\langle (i,j),(k,l) \rangle$ could also be rewired to $\langle (i,l),(j,k) \rangle$, the change in $s$ denoted as $value_{\langle(i,j),(l,k)\rangle}$. Figure \ref{fig:99} illustrates the calculation of the $value$ for a edge pair during the rewiring process.
\begin{figure*}[ht]
%是可选项 h表示的是here在这里插入，t表示的是在页面的顶部插入
\centering
\includegraphics[scale=0.4]{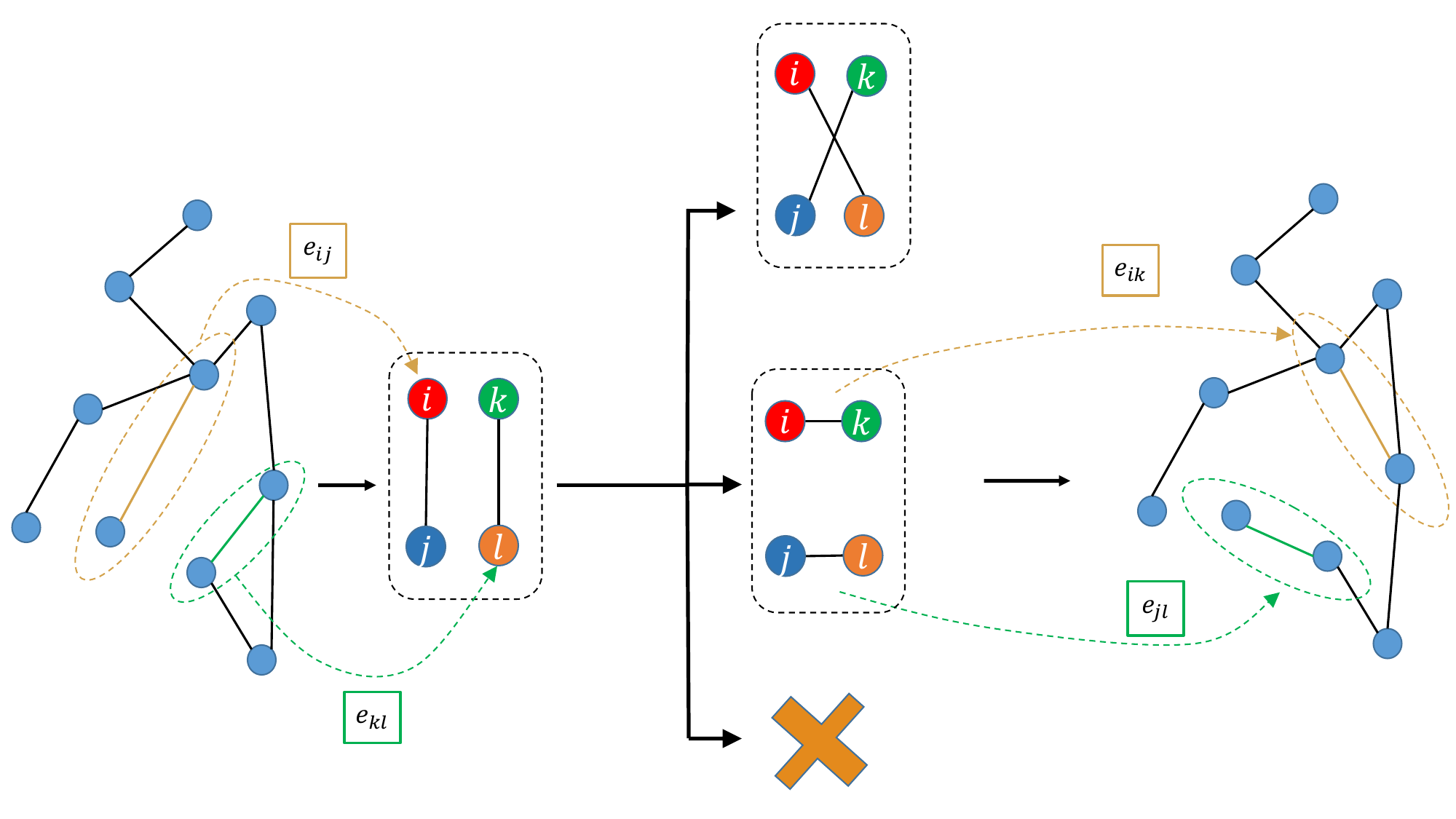}
\caption{The degrees of nodes $i$, $j$, $k$, and $l$ are $4$, $1$, $3$, and $2$, respectively. The rewiring of the edge pairs $\langle(i,j),(k,l)\rangle$ can occur in two possible ways, corresponding to $value_{\{(i,j),(k,l)\}}=(4 \times  3 + 1 \times 2)-(4 \times 1 + 3 \times 2)=4$ and $value_{\{(i,j),(l,k)\}}=(4 \times  2 + 1 \times 3)-(4 \times 1 + 3 \times 2)=1$. If there exist edges $(i, l)$ or $(j, k)$, and $(i, k)$ or $(j, l)$ in the network, then the edge pair $\langle(i, j), (k, l)\rangle$ cannot be rewired.}
\label{fig:99}
\end{figure*}
\subsection{Problem Definition}
For a simple  network $G(V, E)$, let $S$ be the set of rewired edge pairs. We denote the network after rewiring as $G+S$. The assortativity coefficient of $G+S$ is represented by $r(S)$, and the change in the assortativity coefficient can be expressed as $\Delta r(S)$.

In networks, rewiring a limited set of edges to maximize a certain metric is often challenging, as it involves a more complex combinatorial optimization problem compared to adding or removing a limited number of edges to alter a network metric. Here, we assume that newly generated edge pairs resulting from rewiring will not be considered for further rewiring in subsequent steps. This encompasses two scenarios: firstly, if an edge pair $\langle(i,j),(k,l)\rangle$ is reconfigured to $\langle(i,k),(j,l)\rangle$, edges $(i,k)$ and $(j,l)$ will not be rewired with other edges in subsequent steps. Secondly, when edge $(i,j)$ is not rewired, the edge pair $\langle(a,i),(b,j)\rangle$ cannot be rewired to $\langle(a,b),(i,j)\rangle$, because edge $(i,j)$ already exists in the network. However, when edge $(i,j)$ is rewired, the edge pair $\langle(a,i),(b,j)\rangle$ can be reconfigured to $\langle(a,b),(i,j)\rangle$. Nevertheless, our assumption excludes the scenario of considering $\langle(a,i),(b,j)\rangle$ being rewired to $\langle(a,b),(i,j)\rangle$ at any point.
%This includes both the edge pairs formed by the newly created edges along with other existing edges and the edge pairs that were not eligible for rewiring initially but now become eligible. 
Therefore, we can identify all potential edge pairs within the original graph without considering the additional components during the rewiring process. This greatly simplifies our reconfiguration problem. Subsequent experiments can validate the reasonableness of our assumption. 

When rewiring in a network needs to occur in parallel, it is a meaningful assumption that the selected pairs of edges for rewiring align precisely. For instance, in a flight network, continuously adjusting flight routes within a short period is impractical. Instead, the entire flight network typically undergoes a unified adjustment of flight routes at a specific time, necessitating parallel rewiring of flight routes.

We aims to maximize the assortativity coefficient through a limited number of rewirings, name as \textbf{ Maximum Assortative Rewiring (MAR)}. We define the following set function optimization problem: 

% \begin{equation*}
% \begin{split}
% &\max\mathrm{or}\min\quad r(G^{'})\\
% &\mbox{s.t.}\quad  \left\{\begin{array}{lc}
% G^{'} = G(S)\\
% |S|\le k\\
% S\subset EP
% \end{array}\right.
% \end{split}
% \end{equation*}
% =(V,E\setminus S)
\begin{equation}
    \underset{S \subset EP,|S|=k}{maximize} \quad \Delta r(S).
    \label{eq:2}
\end{equation}

where $EP$ is a set of rewirable edges. Since the change in the assortativity coefficient can be converted to the change in $s$, the optimization problem (\ref{eq:2}) is equivalent to the following problem:
\begin{equation}
    \underset{S \subset EP,|S|=k}{maximize} \quad \Delta s(S).
    \label{eq:3}
\end{equation}
In MAR, the set $EP$ consists of all possible rewired edge pairs with a positive $value$ in the original network $G$. These edge pairs in $EP$ satisfy two mutually exclusive conditions. 
\begin{itemize}
    \item Constraint 1: The pair of edges formed by the same edge and other edges are mutually exclusive, as each edge can only be rewired once.
    \item Constraint 2: Edge pairs that result in the same edge after rewiring are also mutually exclusive, since simple graphs do not allow multiple edges between the same pair of nodes. 
\end{itemize}

Figure \ref{fig:1010} illustrates a network along with its corresponding $EP$. Suppose we select the edge pair $\langle(2,3),(4,5)\rangle$ and rewire it to $\langle(2,4),(3,5)\rangle$. According to Constraint 1, the edge pairs $\langle(2,3),(4,5)\rangle$, $\langle(2,8),(4,5)\rangle$, and $\langle(2,3),(6,7)\rangle$ cannot be chosen for the next rewiring process. Following Constraint 2, the edge pair $\langle(2,8),(4,9)\rangle$ also cannot be selected for the next rewiring process.
\begin{figure}[ht]
%是可选项 h表示的是here在这里插入，t表示的是在页面的顶部插入
\centering
\includegraphics[scale=0.3]{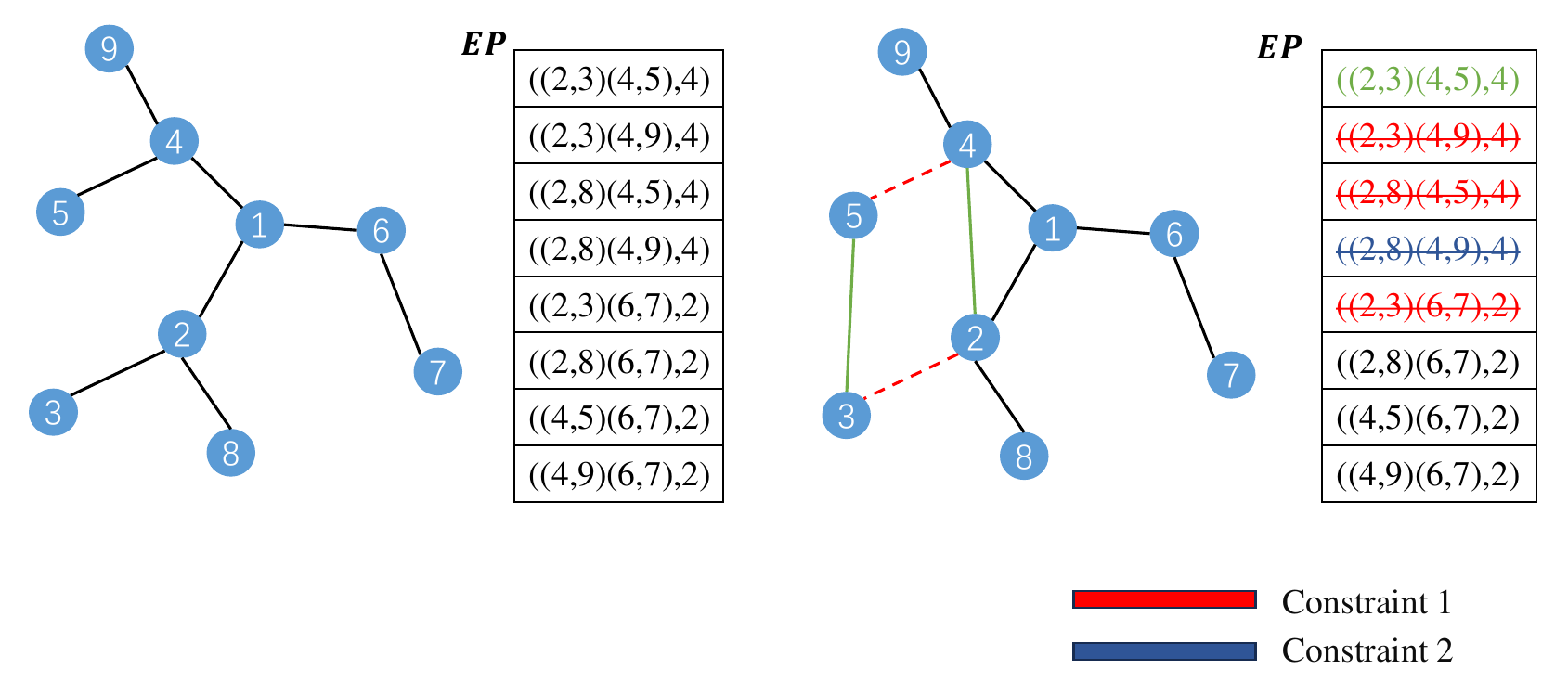}
\caption{The left side illustrates the original network along with its corresponding $EP$. In addition to the rewirable edge pairs, $EP$ also includes their corresponding $value$. The network on the right side represents the change in $EP$ corresponding to the rewiring of the edge pair $\langle(2,3),(4,5)\rangle$ to $\langle(2,4),(3,5)\rangle$. According to Constraint 1, the edge pairs $\langle(2,3),(4,5)\rangle$, $\langle(2,8),(4,5)\rangle$ and $\langle(2,3),(6,7)\rangle$ cannot be chosen for the next rewiring process, we use red lines to indicate this. Following Constraint 2, the edge pair $\langle(2,8),(4,9)\rangle$ also cannot be selected for the next rewiring process, we use orange lines to indicate this.}
\label{fig:1010}
\end{figure}

\begin{theorem}\label{theorem1}
In the MAR problem, $\Delta s(S)$, exhibits monotonic behavior.
\end{theorem}
\begin{proof}
In MAR, for any given solution $S$, let us consider an edge pair $\langle(i, j), (k, l)\rangle$ in $G+S$ that can be rewired. The change in the assortativity coefficient, denoted $\Delta s(S \cup \{\langle(i, j), (k, l)\rangle\})$, can be expressed as $\Delta s(S \cup \{\langle(i, j), (k, l)\rangle\}) = \Delta s(S) + value_{\langle(i, j), (k, l)\rangle}$. Since $value_{\langle(i, j), (k, l)\rangle}>0$, it follows that $\Delta s(S \cup \langle(i, j), (k, l)\rangle) > \Delta s(S)$, indicating that $s(S)$ is increasing monotonically.
\end{proof}
\begin{algorithm}[!t]
  \caption{GA} % 名称
  \begin{algorithmic}[1]
    \Require
      Graph $G=(V,E)$; an integer $k$
    \Ensure
       A set $S$ and $|S|=k$
    % \If{GARS}
        \State $EP \leftarrow $ the set of possible rewired edge pairs with a positive $value$ in the original $G$, sorted in descending order.
    % \EndIf
    % \If{GDRS}
    %     \State $EP \leftarrow $ the set of possible pairs of rewired edges with a negative $value$ in the original $G$, sorted in ascending order.
    % \EndIf
    \State $S \leftarrow \emptyset$
    \State $index \leftarrow 0$
    \State $n \leftarrow 0$
    \State $len \leftarrow length(EP)$
    \While {$n < k$ and $index < len$}
        \State edge $(i,j),(k,l)  \leftarrow EP [index]$
        \State $index \leftarrow index + 1$
        \If{the edges $(i,k)$ and $(j,l)$ can be rewired in $G$}
            \State $S \leftarrow S \cup \{\{(i,j),(k,l)\}\}$
            \State $G \leftarrow G+\{\{(i,j),(k,l)\}\}$
            \State $n \leftarrow n + 1$
        \EndIf
    \EndWhile
    \State
    \Return {$S$}
  \end{algorithmic}
  \label{alg:1}
\end{algorithm}
\begin{theorem}\label{theorem2}
In the MAR problem, $\Delta s(S)$ is submodular.
\end{theorem}

\begin{proof}
For each pair $S$ and $T$ of MAR such that $S \subseteq T$, and for each pair of rewired edge pairs $\langle(i, j), (k, l)\rangle$ in $G(S)$ that satisfy the rewiring requirements, if $\Delta s(S)$ is submodular, then $s(S \cup \{\langle(i, j), (k, l)\rangle\}) - s(S)$ should be greater than or equal to $s(T \cup \{\langle(i, j), (k, l)\rangle\})-s(T)$. We know that the impact of rewiring a pair of edges on the network's assortativity coefficient only depends on that specific pair of edges, and rewiring other pairs of edges will not affect the assortativity coefficient change of this specific pair. so $s(S \cup \{\langle(i, j), (k, l)\rangle\}) - s(S) = s(T \cup \{\langle(i, j), (k, l)\rangle\})-s(T) = value_{(i, j), (k, l)}$, so $\Delta s(S)$ is submodular.
\end{proof}

\subsection{Rewiring Method}
Let's consider the following optimization problem: given a finite set $N$, an integer $k$, and a real-valued function $z$ on the set of subsets of $N$, find a set $S \in N$ with $|S| \leq k$ such that $z(S)$ is maximized. If $z$ is monotone and submodular, the following greedy algorithm achieves an approximation of $1-\frac{1}{e}$ \cite{nemhauser1978analysis}: start with the empty set and repeatedly add the element that maximizes the increase in $z$ when added to the set. Theorem~\ref{theorem1} and~\ref{theorem2} indicate that the objective function (\ref{eq:3}) is both monotone and submodular. As a result, a simple greedy strategy can be used to approximate the problem (\ref{eq:2}). We propose the \textbf{Greedy Assortative} to maximize the assortative coefficient.

\textbf{Greedy Assortative(GA):}
First, identify all possible pairs of rewired edges with a positive $value$ in the original graph $G$. Initialize the set $S$ is empty. Then select the pair with the highest positive $value$ and try to rewire it. If successful, add it to $S$. if not, move on to the pair with the second highest $value$ and repeat the process until $|S|=k$.

The details of this algorithm are summarized in Algorithm \ref{alg:1}. In fact, the time complexity of the algorithm is $O(M^3\log(M))$, where $M$ represents the number of edges in the graph. The GA method requires identifying all possible rewiring edge pairs with positive $value$ and sorting them in descending order. When the size of a network is large, the number of potential edge pairs is enormous, and the primary time cost of the algorithm lies in sorting these large numbers of potential edge pairs. Although there are sorting algorithms available that can effectively reduce sorting time, it may still be time-consuming for a large-scale network. Indeed, there is relatively little research on changing network degree correlations through a limited number of rewirings, and there are few related heuristic rewiring methods available at present. Therefore, considering the characteristics of assortative networks, we propose several heuristic methods with a time complexity of $O(N)$ or $O(N^2)$.

\textbf{Edge Difference Assortative(EDA):}
To enhance network assortativity, we prioritize rewiring edges with a large difference in degrees between their endpoints. In the rewiring process, we first select the edge with the largest difference in degrees, then proceed to choose the edge with the next largest difference in degrees that satisfies the rewiring condition. This selected edge pair is then rewired to ensure that the edge with the largest difference in degrees is addressed. We continue this process by selecting the edge with the largest difference in degrees from the remaining edges.

\textbf{Targeted Assortative(TA):}
This is an adaptation of Geng's disassortative rewiring strategy\cite{geng2021global}, which prioritizes connecting nodes with higher degrees to nodes with lower degrees, thereby inducing disassortativity in the network. We employ a similar approach, giving priority to rewiring that connects nodes with the highest degrees before considering connections among other nodes.

\textbf{Probability Edge Assortative(PEA):} 
Probability assortative considers the tendency of high-degree nodes to connect, enhancing network assortativity. We can further enhance assortativity by focusing on rewiring edges with a significant difference in degrees. Initially, calculate the degree difference for each edge in the network, using the degree difference as the probability weight for edge selection. Probabilistically choose two edges, disconnect them, and then connect the high-degree nodes with each other and the low-degree nodes with each other.

Next, we focus on explaining more implementation details of the three heuristic methods we proposed or improved.

The EDA algorithm, as shown in Algorithm \ref{alg:567}, first sorts the edges in the network in descending order based on the degree difference. It selects the edge with the largest degree difference, denoted as $(i, j)$, and then attempts to rewire it with the edge with the second largest degree difference, denoted as $(k, l)$. We then sort the four nodes corresponding to these two edges in descending order of their degrees, denoted as $a \geq b \geq c \geq d$. We rewire the edge pair $\langle(i, j),(k, l)\rangle$ to $\langle(a, b)(c, d)\rangle$, thereby disconnecting nodes with large degree differences while connecting nodes with similar degrees, thus enhancing the network's assortativity. If rewiring is not possible, we proceed to select the next edge in the sequence and attempt to rewire it. If none of the edges can be rewired with it, the edge is removed from the sequence.

The TA algorithm, as shown in Algorithm \ref{alg:2}, utilizes a $nodeList$, which is a list of all nodes in the network arranged in descending order of their degrees. Node $a$ represents the highest degree node in each primary iteration, while node $z$ represents the next highest degree node which has not been rewired yet in each primary iteration. $p$ and $q$ represent the indices of nodes $a$ and $z$ in the $nodeList$, respectively. $S(a)$ denotes the set of neighbor nodes of node $a$, while $S(a)-S(y)$ represents the set of nodes that are neighbors of node a but not neighbors of node $y$. Node $y$ is the node with the minimum degree in the set $S(z)$, and node $b$ is the node with the minimum degree in the set $S(a)-S(y)$. The degrees $d_z$, $d_y$, and $d_b$ are defined similarly. The condition $d_z > d_y$ and $d_z > d_b$ indicates that reconnecting the edge pair $\langle(a, b), (z, y)\rangle$ to $\langle(a, z), (b, y)\rangle$ effectively enhances the network's assortativity. The terminal condition of the algorithm is not solely determined by the budget $k$. When the budget $k$ is large or when the network size is small, the algorithm may terminate before reconnecting $k$ times due to constraints such as $d_z > d_y$ and $d_z > d_b$, indicating termination after considering all nodes.

The PEA algorithm, as shown in Algorithm \ref{alg:3}, first calculates the degree difference for each edge pair of nodes, denoted as $D_k = [diff_1, diff_2, diff_3, ..., diff_M]$. We can compute the probability density for each edge as $p_i = d_i / \sum(N_k)$. Based on the probabilities $P_k$, we select the edge pair $\langle(i, j), (k, l)\rangle$, where edges with larger degree differences have a higher probability of being chosen. The rewiring process corresponds to that in EDA.
\begin{algorithm}[!t]
  \caption{EDA} % 名称
  \begin{algorithmic}[1]
    \Require
      Graph $G=(V,E)$; an integer $k$.
    \State $n \leftarrow 0$
    \State $edgeList \leftarrow $ A list of edge in $G$.
    \While {$n < k$}
        \State The $edgelist$ sorted in descending order based on the degree difference.
        \State $(i,j) \leftarrow edgeList[0]$
        \State $p \leftarrow 1$ the degree of $a$
        \While{$p < length(edgeList)$}
            \State $(k,l) \leftarrow edgeList[p]$
            \State $a,b,c,d \leftarrow$ The nodes of the two edges $(i, j)$ and $(k, l)$ are arranged in descending order based on their degrees.
            \If{$(i,j),(k,l)$ can be rewired to $(a,b),(c,d)$}
                \State $G \leftarrow G+\{\langle (a,c),(b,d)\rangle\}$
                \State $n \leftarrow n + 1$
                \State $edgeList \leftarrow edgeList - \{(a,b),(c,d)\} + \{(a,c),(b,d)\}$
                \State $n \leftarrow n + 1$
            \Else
                \State $p \leftarrow p + 1$
                \If{$p == length(edgeList)$}
                    \State $edgeList \leftarrow edgeList - \{(i,j)\}$
                \EndIf
            \EndIf
        \EndWhile
    \EndWhile
  \end{algorithmic}
  \label{alg:567}
\end{algorithm}
\begin{algorithm}[!t]
  \caption{TA} % 名称
  \begin{algorithmic}[1]
    \Require
      Graph $G=(V,E)$; an integer $k$.
    % \Ensure
    %    A set $S$ and $|S|=k$
    % \If{GARS}
    \State $nodeList \leftarrow $ A list of nodes sorted in descending order based on node degree.
    \State $n \leftarrow 0$
    \State $p \leftarrow 0$
    \State $q \leftarrow p + 1$
    \State $N \leftarrow length(nodeList)$
    \While {$n < k$ and $p < N - 1$}
        \If{$q = N$}
            \State $p \leftarrow p + 1$
            \State $q \leftarrow p + 1$
            \State continue
        \EndIf
        \State Get the node with highest degree as $a$ according to $nodeList[p]$
        \State $d_a \leftarrow$ the degree of $a$
        \State Get the node with lowest degree as $z$ according to $nodeList[q]$
        \State $d_z \leftarrow$ the degree of $z$
        \State $key \leftarrow True$
        \While{The $G$ has the edge $(a,x)$}
            \State $q \leftarrow q + 1$
            \If{$q = N$}
                \State $key \leftarrow False$
                \State break
            \EndIf
            \State $z \leftarrow nodeList[q]$
            \State $d_z \leftarrow$ the degree of $z$
            \If{$key = False$}
                \State $p \leftarrow p + 1$
                \State $q = p + 1$
            \Else
                \State $S_a \leftarrow$ the neighbors nodes of $a$
                \State $S_z \leftarrow$ the neighbors nodes of $z$
                \State the node $y$, which degree smallest in $S_z$
                \State $S_y \leftarrow$ the neighbors nodes of$ y$
                \State $S_{a-y} \leftarrow S_a - S_y$  
                \If{$S_{a-y} = \emptyset$}
                    $q = q + 1$
                \Else
                    \State the node $b$, which degree smallest in $S_{a-y}$
                    \If{$d_z > d_y$ and $d_z > d_b$}
                        \State $G \leftarrow G+\{\langle (a,b),(z,y)\rangle\}$
                        \State $n \leftarrow n + 1$
                        \State $q \leftarrow q + 1$
                    \Else
                        \State $q \leftarrow q + 1$
                    \EndIf
                \EndIf
            \EndIf
        \EndWhile
    \EndWhile
  \end{algorithmic}
  \label{alg:2}
\end{algorithm}
\begin{algorithm}[!t]
  \caption{PEA} % 名称
  \begin{algorithmic}[1]
    \Require
      Graph $G=(V,E)$; an integer $k$.
    \State $D_k \leftarrow [diff_1,diff_2,diff_3,...,diff_M] $, the difference in degrees between the nodes at both ends of each edge.
    \State $P \leftarrow $ A probability distribution is calculated for each edge based on the difference in degrees of the two end nodes.
    \State $n \leftarrow 0$
    \While {$n < k$}
        \State $(i,j),(k,l) \leftarrow$ Randomly select two edges based on the probability distribution $P$.
        \State $a,b,c,d \leftarrow$ The nodes of the two edges $(i, j)$ and $(k, l)$ are arranged in descending order based on their degrees.
        \If{$(i,j),(k,l)$ can be rewired to $(a,b),(c,d)$}
            \State $G \leftarrow G+\{\langle (a,c),(b,d)\rangle\}$
            \State $n \leftarrow n + 1$
        \EndIf
    \EndWhile
  \end{algorithmic}
  \label{alg:3}
\end{algorithm}

\subsection{Network Robustness}
\label{nr}
Robustness refers to the ability of a network to continue operating and supporting its services when parts of the network are naturally damaged or subjected to attacks. For example, in a power network, a robust electrical network should continue functioning without significant impact even if some power plants are unable to operate or certain lines are disrupted. There are currently many robustness metrics available to measure the robustness of a network. Different robustness metrics have different implications for the robustness of a network. For example, the average shortest path\cite{motter2002cascade,morohosi2010measuring} and efficiency\cite{pu2012efficient,schieber2016information} quantify the shortest path distances between pairs of nodes in the network. $f$-robustness\cite{jing2007effects} and $R$-robustness\cite{louzada2013smart,herrmann2011onion} are directly related to the largest connected component of the network. In addition to these metrics that utilize the network's topology to quantify its robustness, there exists another type of robustness metric based on the adjacency matrix, known as spectral-based robustness metrics. Spectral-based robustness metrics have been demonstrated to be associated with information propagation and dynamic processes in networks, and as such, they are widely utilized for measuring network robustness. There is existing research suggesting a certain relationship between degree correlation and network robustness. In this study, we primarily investigate whether our rewiring strategy, aimed at enhancing network degree correlation, can simultaneously improve network robustness. We focus mainly on robustness metrics based on the adjacency matrix.

We consider three adjacency matrix-based robustness metrics, including spectral radius and natural connectivity.
\begin{enumerate}
\item Spectral radius\cite{tong2010vulnerability}: The spectral radius, denoted as $\lambda_1$
 , of a network is defined as the largest eigenvalue of the network's adjacency matrix.
% \item Spectral gap: The spectral gap is defined as the difference between the largest and the second largest eigenvalue of the adjacency matrix, denoted as $\lambda_1 - \lambda_2$.
\item Natural connectivity\cite{jun2010natural}: The natural connectivity is a mathematical measure defined as a special average of all the eigenvalues of the adjacency matrix with respect to the natural exponent and natural logarithm. It is directly related to the closed paths in the network. This metric is defined as:
\begin{equation}
    \bar{\lambda}(G) = ln(\frac{1}{n}\sum_{i=1}^{n}e^{\lambda_i}).
    \label{eq:08}
\end{equation}
\end{enumerate}

\subsection{Robustness of Centrality Measures}
\label{rocm}
\subsubsection{Centrality Measures}
Centrality measures are a method used to assess the importance of nodes in a network, commonly used in the study of complex networks such as social networks, information diffusion networks, transportation networks, and more \cite{kitsak2010identification}. We are interested in whether the centrality measures of the network are robust when we use our rewiring method to enhance the degree correlation of the network. We consider four widely applied centrality metrics: betweenness centrality, closeness centrality, eigenvector centrality, and k-shell.

Betweenness centrality measures the importance of a node in a network based on the number of shortest paths that pass through it\cite{freeman1977set}; Closeness centrality measures the average distance between a node and all other nodes in a network\cite{krackhardt1990assessing}; Eigenvector centrality measures the importance of a node in a network, taking into account both the node's own influence on the network and the influence of its neighboring nodes\cite{bonacich2007some}; The k-shell method calculates the node centrality by decomposing the network\cite{carmi2007model}.

\subsubsection{Robustness evaluation function of centrality measures}
As the network topology changes with the rewiring, the degree correlation of the network also changes, but the degree sequence of the network remains unchanged. This prompts us to investigate whether different centrality measures of the network exhibit robustness under rewiring strategies aimed at enhancing network degree correlation.

To evaluate the robustness of centrality measures $C$, we calculate the Spearman rank correlation coefficient $SC$ between the centrality measures $C_O$ and $C_R$ before and after rewiring, respectively. $C_O$ represents the centrality measure of the original network, while $C_R$ represents the centrality measure of the rewired network. Here, we represent the node rankings corresponding to $C_O$ and $C_R$ as $R_O$ and $R_R$, respectively. The Spearman rank correlation coefficient $SC$ can be calculated as follows:
\begin{equation}
\mathbf{SC} = \frac{\langle R_OR_R \rangle - \langle R_O \rangle \langle R_R \rangle}{\sqrt{(\langle R_O^2 \rangle - \langle R_O \rangle^2)(\langle R_M^2 \rangle - \langle R_M \rangle^2)}}
\end{equation}

% A higher value of $SC$ indicates a more robust centrality measure. 

The value of $SC$ ranges from -1 to 1, with a value closer to 1 indicating robustness for the respective centrality measure.
% In this paper, we mainly focus on five popular centrality measures, namely betweenness centrality, closeness centrality, information centrality, eigenvector centrality, and PageRank centrality.

%% 写算法伪代码或者流程的前期准备
\renewcommand{\algorithmicrequire}{\textbf{Input:}}  % Use Input in the format of Algorithm
\renewcommand{\algorithmicensure}{\textbf{Output:}} % Use Output in the format of Algorithm

\section{Experiments}\label{thre}
In this section, we first demonstrate the reasonableness of our assumptions and compare the GA method with the optimal solution. We validate the effectiveness of the GA method and our heuristic methods on real networks and explore their impact on network spectral robustness metrics. Finally, we investigate whether various centrality measures can maintain robustness during network rewiring using the GA method.

\subsection{Baseline Method}\label{bas}
Currently, there are limited methods for altering the assortativity coefficient of a network through degree-preserving rewiring. To demonstrate the effectiveness of our proposed GA method and three heuristic methods, we compare them with the following two existing heuristic methods.
\begin{enumerate}
\item Random Assortative(RA)\cite{xulvi2005changing}: Randomly select two edges without common nodes. Rewire these edges so that the two highest degree node and the two lowest-degree nodes are connected.
\item Probability Assortative(PA)\cite{jing2016algorithm}: The probability of selecting a node is determined by its degree, serving as a probability weight. The process involves probabilistically choosing two nodes, $i$ and $k$, and then selecting random neighbors, $j$ and $l$, for nodes $i$ and $k$, respectively. These chosen nodes form the rewired edges $(i, j)$ and $(k, l)$, resulting in their disconnection, followed by the connection of edges $(i, k)$ and $(j, l)$.

\end{enumerate}

Both of these algorithms are relatively simple, and their specific procedures are detailed in their corresponding papers; therefore, we will not provide a detailed description here.

\subsection{Dataset description}
We evaluate the methods using three different categories of datasets, as indicated in Table \ref{tab1}. These categories include AS router, flight, and power networks. Edge rewiring in these networks holds practical significance and applications. For Instance, in the flight network, edge rewiring involves rearranging flights between airports without affecting the airport's capacity.
\begin{itemize}
    \item \textbf{AS-733}\cite{leskovec2005graphs} The dataset consists of routing networks spanning $733$ consecutive dates. In our experiments, we selected a routing network every six months, resulting in a total of six networks. The size of the networks gradually increased, with the number of nodes ranging from 3015 to 6127, and the number of edges ranging from 5156 to 12046. All these networks are disassortative scale-free networks with degree exponent between 2 and 3.
    \item \textbf{USPowerGrid and BCSPWR10}\cite{kunegis2013konect,rossi2015network} These are two power networks for the Western states of the United States, both of which belong to neutral networks. And the degree distribution of the power network follows an exponential distribution.
    \item \textbf{USAir97 and USAir10}\cite{kunegis2013konect,rossi2015network} The USAir97 and USAir10 are flight networks composed of the air routes between American airports in 1997 and 2010, respectively. The degree distributions of these two networks lie between exponential and power-law distributions, often referred to as stretched exponential distributions.
\end{itemize}
\begin{table}[ht]
\begin{center}
\caption{Statistics of datasets. 3 categories of datasets (AS router, power, and flight networks) where rewiring can be applied. For a network with $\lvert $V$ \rvert$ nodes and $\lvert $E$ \rvert$ edges, we use $r$ to denote the assortativity coefficient of the network.}
\label{tab1}
\begin{tabular*}{\hsize}{@{}@{\extracolsep{\fill}}cccc@{}}
\hline
 % & & & & \\[-6pt]
Dataset & $\lvert $V$ \rvert$ & $\lvert $E$ \rvert$  & $r$
\\ \hline 
AS-733-A   & 3015 & 5156& -0.229 
\\ 
AS-733-B   & 3640 & 6613& -0.210
\\ 
AS-733-C   & 4296 & 7815& -0.201
\\ 
AS-733-D   & 5031 & 9664& -0.187
\\ 
AS-733-E   & 6127 & 12046& -0.182
\\ \hline
USPowerGrid   & 4941 & 6594& 0.003 
\\ 
BCSPWR10   & 5300 & 8271& -0.052
\\ \hline
USAir97   & 332 & 2126& -0.208
\\ 
USAir10   & 1574 & 17215&-0.113
\\ \hline
\end{tabular*}
\end{center}
\end{table} 
\subsection{Assumption rationality}
We assume that during the rewiring process, newly generated edge pairs will not be rewired in subsequent steps. Below, we aim to verify the reasonableness of this assumption. Even for a small-scale network, enumerating all possible rewiring edge pairs to find the optimal solution for rewiring k edge pairs is challenging. Therefore, our goal is to validate whether our GA method can approach the maximum assortativity achievable by the network under this assumption. If, under our assumption, the GA method can bring the network close to maximum assortativity, it indicates that our assumption does not significantly affect the rewiring effectiveness, thereby validating its reasonableness.

Winterbach~\emph{et al.}\cite{winterbach2012greedy} investigated an exact approach to obtain the maximum assortative network that can be formed with a given degree sequence. They transformed the problem of constructing the maximum assortative network into the maximum weight subgraph problem on a complete graph, which was solved using b-matching \cite{papadimitriou1998combinatorial}. Furthermore, they further converted b-matching into a more efficient 1-matching problem \cite{shiloach1981another} to obtain the maximum assortative network for a given degree sequence. Considering that the time complexity of 1-matching is also relatively high, we conducted experiments on three small-scale synthetic networks. In the experiments, we first obtained the maximum assortative network achievable with the degree sequence using Winterbach~\emph{et al.}'s method and then executed the GA method to obtain the maximum assortative network. We compared whether the assortativity coefficient of the maximum assortative network obtained by the GA method could match that of the maximum assortative network obtained using Winterbach~\emph{et al.}'s method to assess the reasonableness of the assumption.

The experimental results are summarized in Table \ref{tab456}, where we present the maximum, minimum, and average approximation ratios of the assortativity coefficients obtained by the GA method compared to the theoretically maximum assortative networks across various types of networks. In the case of the WS network, the minimum approximation ratio is 0.927 and the average approximation ratio is 0.984. For the other two types of networks, the minimum and average approximation ratios are better than those of the WS network. This suggests that even under our assumption, our GA method can effectively approximate the maximum assortativity coefficient across all three types of networks. When our goal is to maximize the assortativity coefficient by rewiring a limited number of edge pairs, our algorithm typically performs better because it is less likely to select newly created edge pairs during the rewiring process compared to obtaining the network's maximum assortative network.

\begin{table}[htp]
\begin{center}
\caption{Comparing the assortativity coefficient of the maximum assortative network obtained by the GA method and the exact approach on three model networks. The first three columns denote the network type, number of nodes, and number of edges in the network. The fourth column indicates the maximum approximation ratio achieved by GA, while the fifth column presents the minimum approximation ratio achieved by GA. The sixth column displays the average approximation ratio.}
\label{tab456}
\begin{tabular*}{\hsize}{@{}@{\extracolsep{\fill}}cccccc@{}}
\hline
Network & $\lvert $V$ \rvert$ & $\lvert $E$ \rvert$  & Max Approx. & Min Approx. &Ave Approx.
\\ \hline
ER   & 50 & 100& 0.990&0.932&0.968
\\ 
WS   & 50 & 100& 1 &0.927&0.964
\\ 
BA   & 50 & 96& 0.997 &0.957&0.982
\\ 
\hline
\end{tabular*}
\end{center}
\end{table} 

\subsection{Solution Quality}
In this section, we first formulate the Integer Programming(IP) for MAI to obtain the optimal solution. We validate the effectiveness of GA on several small model networks,ER network, WS network and BA network. Subsequently, using the real networks from Table \ref{tab1}, we compare GA with baseline methods introduced in \ref{bas}, confirming the effectiveness of GA across different types of real networks. Finally, we analyze the runtime of GA on real networks.
\subsubsection{IP formulation for MAI}
Let $S$ be a solution for MAI, and $EP$ represent all pairs of edges in the network that can be rewired, each with a positive value. Given each edge pair $ep \in EP $, we define $x_{ep}$
\begin{equation*}
\begin{split}
x_{ep}=\left\{\begin{array}{lc}
1 \,\,\,\,\text{if}\, ep \in S\\
0 \,\,\,\,\text{otherwise.}\\
\end{array}\right.
\end{split}
\end{equation*}
The IP formulation is defined as follows:
\begin{equation*}
\begin{split}
&\max \,\,\sum_{ep \in EP}{value_{ep}x_{ep}} \\
&\text{s.t.}\quad  \left\{\begin{array}{lc}
\sum_{\{ep\in EP 
| (i,j) \in ep\}}
x_{ep} \leq 1 \,\,\,\text{for\,each} \, \,(i,j) \in E\\
\sum_{\{ep\in EP | (i,j) \in ep_{r}\}}
x_{ep}\leq 1 \,\,\,\text{for\,each}\,\, (i,j) \in E_{r}\\
\sum_{ep \in EP} x_{ep} \leq k
\\
x_{ep} \in \{0,1\} \,\,\,\text{for\,each}\,\,ep \in EP \\
\end{array}\right.
\end{split}
\end{equation*}
$E_r$ is a set of new edges generated after rewiring the elements in $EP$, and $ep_r$ represents the edge pair generated after rewiring $ep$. The first constraint ensures that each edge in the original network can only be rewired once. The second constraint ensures that each new edge is only generated once. 

We solved the above program by using the GLPK solver. In the experiment, we compared GA and the optimal solution calculated using IP. Our experiments are conducted on three popular model networks: ER network, WS network, and BA network. Since these networks are randomly generated, we repeat the experiments multiple times and average the results. In the experiments, we consider the rewiring frequency to be 5\% of the network edges. 

\begin{table}[htbp]
\begin{center}
\caption{Comparing GA and the optimal solution on three model networks. The first three columns denote the network type, number of nodes, and number of edges in the network. The fourth column represents the percentage of times GA obtains an optimal solution in multiple experiments. The fifth column indicates the minimum approximation ratio achieved by GA, while the sixth column presents the average approximation ratio.}
\label{tab2}
\begin{tabular*}{\hsize}{@{}@{\extracolsep{\fill}}cccccc@{}}
\hline
Network & $\lvert $V$ \rvert$ & $\lvert $E$ \rvert$  & OPT\% & Min Approx. &Ave Approx.
\\ \hline
ER   & 50 & 100& 42.5&0.960&0.960
\\ 
WS   & 50 & 100& 67.0 &0.924&0.990
\\ 
BA   & 50 & 96& 99.5 &0.994&0.999
\\ 
\hline
\end{tabular*}
\end{center}
\end{table} 

% \begin{table}[ht]
% \begin{center}
% \caption{Robustness improvement on all our real-world networks after the rewiring budget is 5\% of the number of network edges. $\lambda_{1}$: spectral radius, $\lambda_1 - \lambda_2$: spectral gap, and $N$: natural connectivity.}
% \label{tab3}
% \begin{tabular*}{\hsize}{@{}@{\extracolsep{\fill}}llcccccr@{}}
% \hline
%  & & & & \\[-6pt]
% Dataset &state& $\lambda_{1}$ & $\lambda_1 - \lambda_2$ & $N$
% \\ \hline 
% \multirow{2}{*}{AS-733-A}   & original & 31.75& 11.67&23.74
% \\
%    & rewired & 35.82& 18.10&27.81
% \\ \hline
% \multirow{2}{*}{AS-733-B}   & original & 35.54& 15.04&27.36
% \\ 
%    & rewired & 40.92& 21.42&32.74
% \\ \hline
% \multirow{2}{*}{US Power Grid}  & original & 7.48& 0.87&1.46
% \\ 
%   & rewired & 12.12& 4.37&3.73
% \\ \hline
% \multirow{2}{*}{Bus-685}   & original & 6.79& 0.59&2.14
% \\
%    & rewired & 7.45& 1.20&2.21
% \\ \hline
% \multirow{2}{*}{USAir97}  & original & 41.23& 23.92&35.43
% \\
%  & rewired & 43.73& 29.37&37.92
% \\ \hline
% \multirow{2}{*}{USAir500}  & original & 48.07& 29.34&41.86
% \\
%  & rewired & 51.24& 36.24&45.03
% \\ \hline
% \end{tabular*}
% \end{center}
% \end{table} 

The results are reported in Table \ref{tab2}, where we display the percentage of optimal solutions achieved by GA, along with the minimum (i.e., worst-case) and average approximation ratios. The experiments clearly indicate that the minimum approximation ratio achieved by GA significantly outperforms theoretical values. In the BA network, GA obtains an optimal solution in over 99.5\%. Although in ER and WS networks, GA achieves an optimal solution in 42.5\% and 67.0\%, respectively, by observing their minimum and average approximation ratios, it is evident that even when GA does not achieve the optimal solution, it comes very close. For example, in the ER network, the minimum approximation ratio is 0.924, and the average approximation ratio is 0.990. For the three model networks mentioned above, the minimum approximation ratio is not less than 0.924, and the average approximation ratio is not less than 0.960, indicating that GA performs exceptionally well on model networks.

\begin{figure*}[ht]%调节图片位置，h：浮动；t：顶部；b:底部；p：当前位置]
	\centering
	\subfloat[AS-733-A]{
        \begin{minipage}[b]{.3\linewidth}
            \centering
            \includegraphics[scale=0.3]{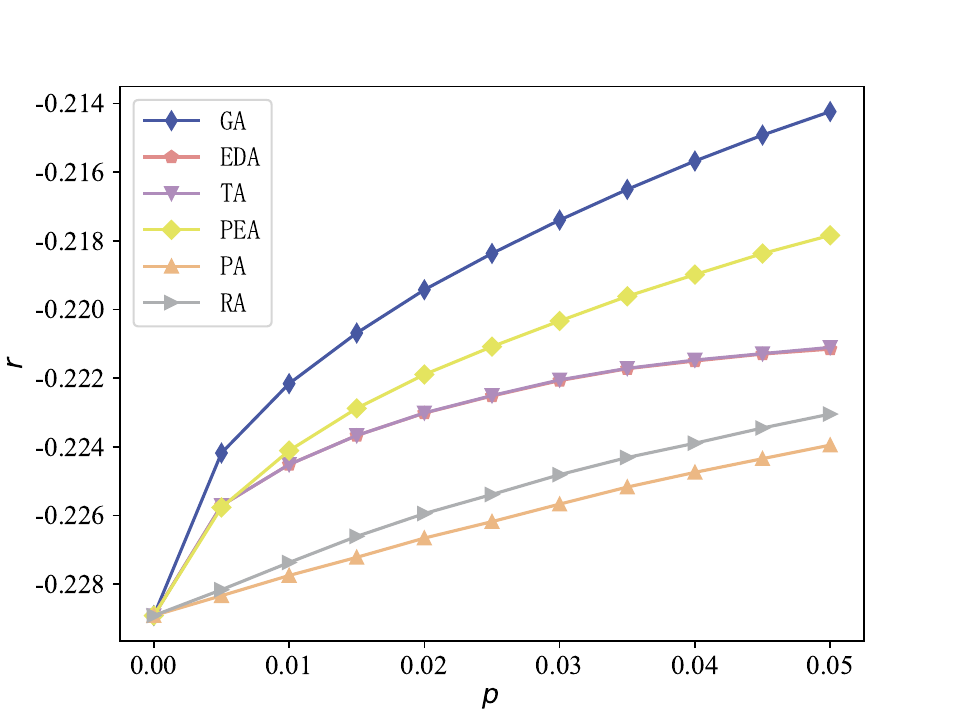}
        \end{minipage}
    }
  % \hspace{2cm}
    \subfloat[USPowerGrid]{
        \begin{minipage}[b]{.3\linewidth}
            \centering
            \includegraphics[scale=0.3]{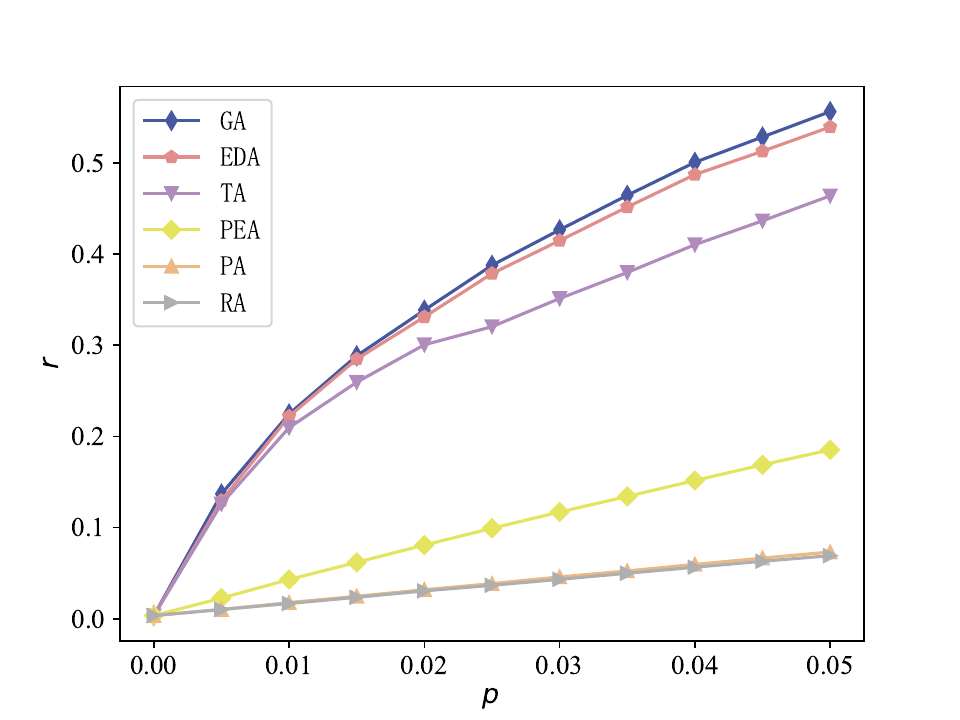}
        \end{minipage}
    }
      % \hspace{2cm}
	\subfloat[USAir97]{
        \begin{minipage}[b]{.3\linewidth}
            \centering
            \includegraphics[scale=0.3]{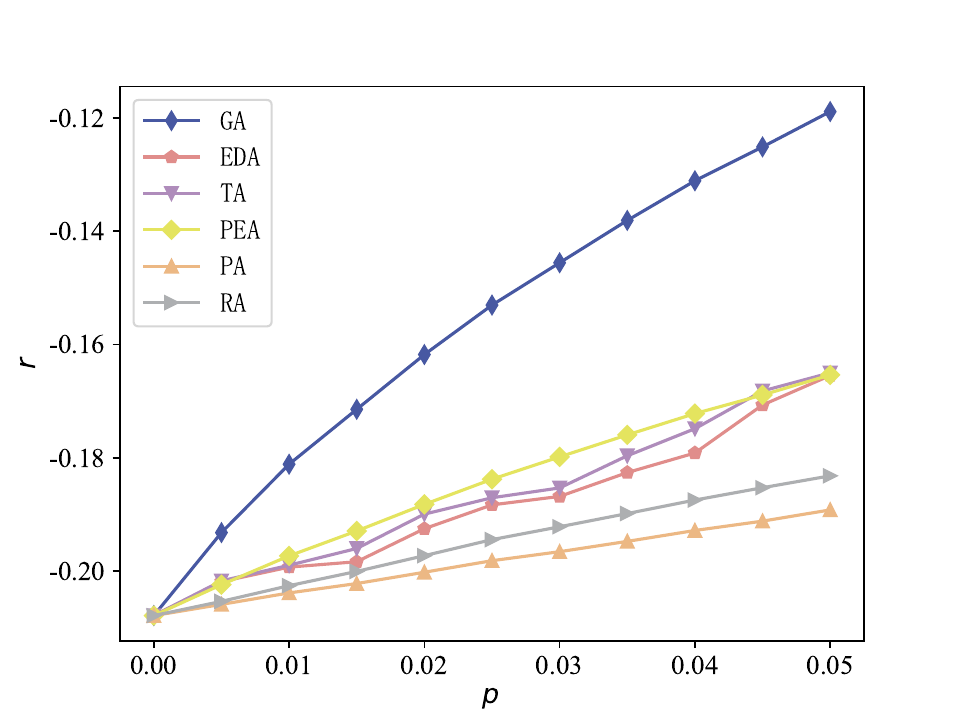}
        \end{minipage}
    }
  \vspace{-0.3cm}
	\\
   % \hspace{-0.5cm}
	\subfloat[AS-733-E]{
        \begin{minipage}[b]{.3\linewidth}
            \centering
            \includegraphics[scale=0.3]{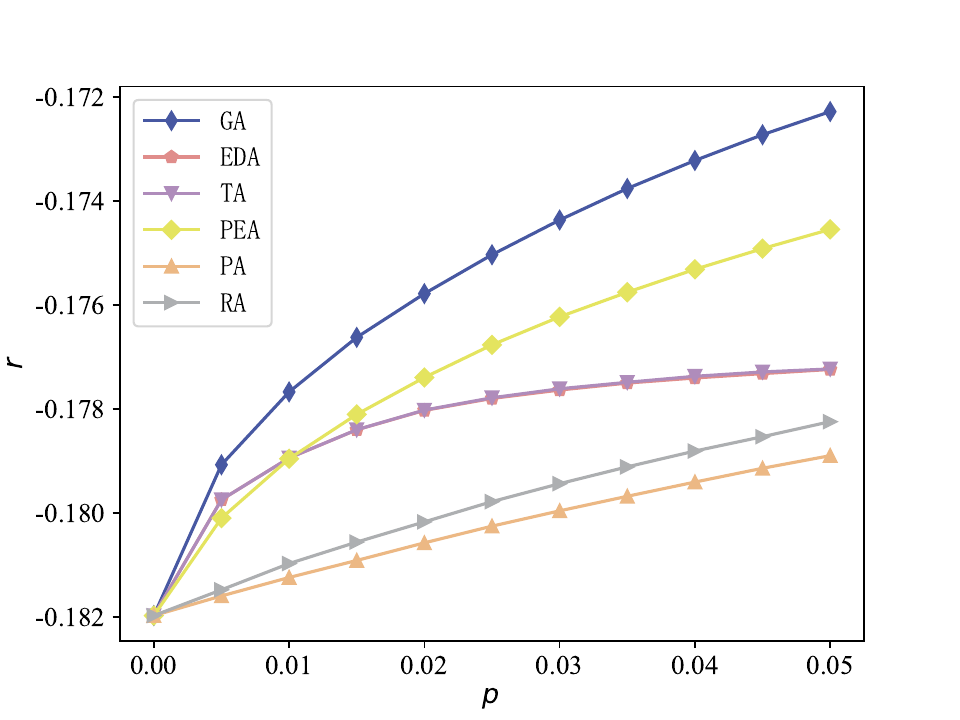}
        \end{minipage}
    }
	\subfloat[BCSPWR10]{
        \begin{minipage}[b]{.3\linewidth}
            \centering
            \includegraphics[scale=0.3]{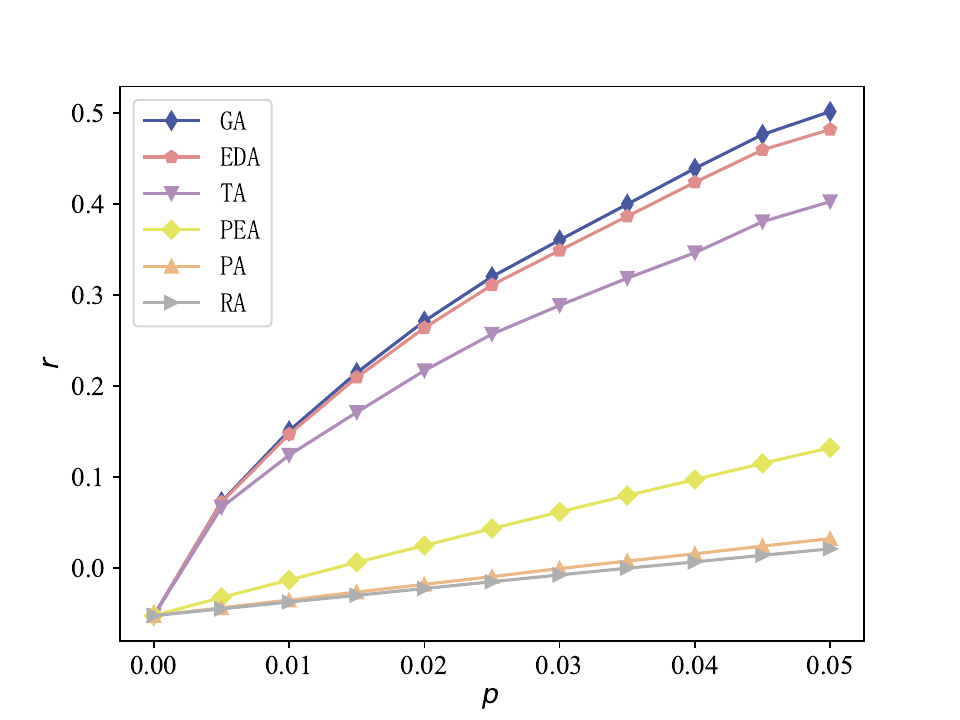}
        \end{minipage}
    }
  % \hspace{2cm}
    \subfloat[USAir10]{
        \begin{minipage}[b]{.3\linewidth}
            \centering
            \includegraphics[scale=0.3]{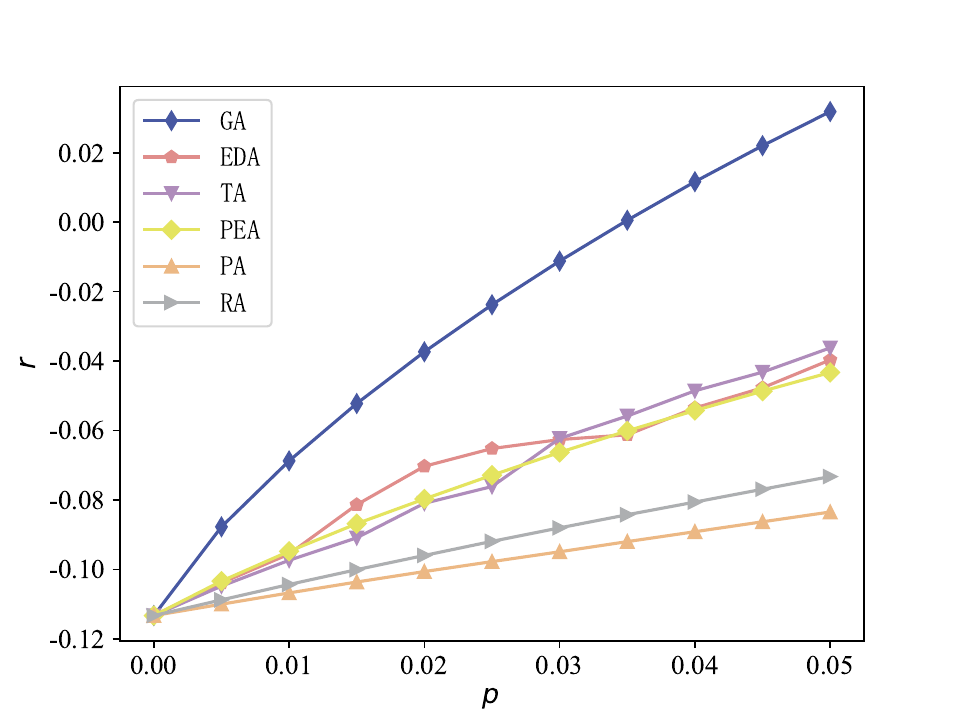}
        \end{minipage}	
    }
	\caption{The assortativity coefficient of the pivot as a function of the percentage $p$ of rewired edge pairs is examined using six methods.}
	\label{fig:1} 
\end{figure*}
\begin{figure*}[!t]%调节图片位置，h：浮动；t：顶部；b:底部；p：当前位置]
	\centering
	\subfloat[AS-733-E]{
        \begin{minipage}[b]{.3\linewidth}
            \centering
            \includegraphics[scale=0.3]{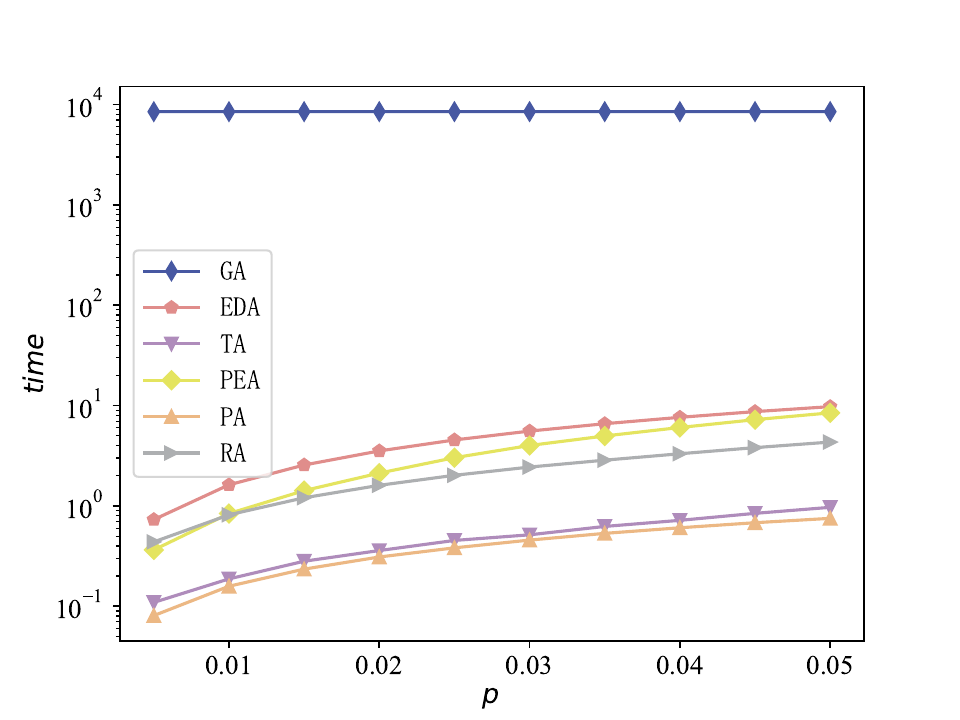}
        \end{minipage}
    }
  % \hspace{-0.5cm}
	\subfloat[USPowerGrid]{
        \begin{minipage}[b]{.3\linewidth}
            \centering
            \includegraphics[scale=0.3]{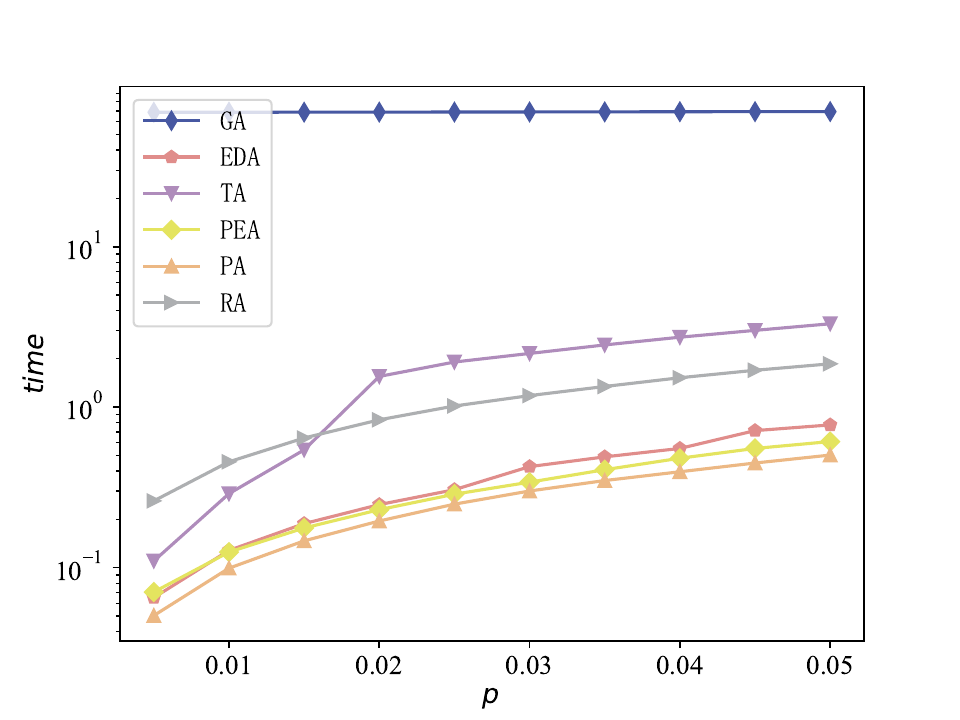}
        \end{minipage}
    }
  % \hspace{2cm}
    \subfloat[USAir97]{
        \begin{minipage}[b]{.3\linewidth}
            \centering
            \includegraphics[scale=0.3]{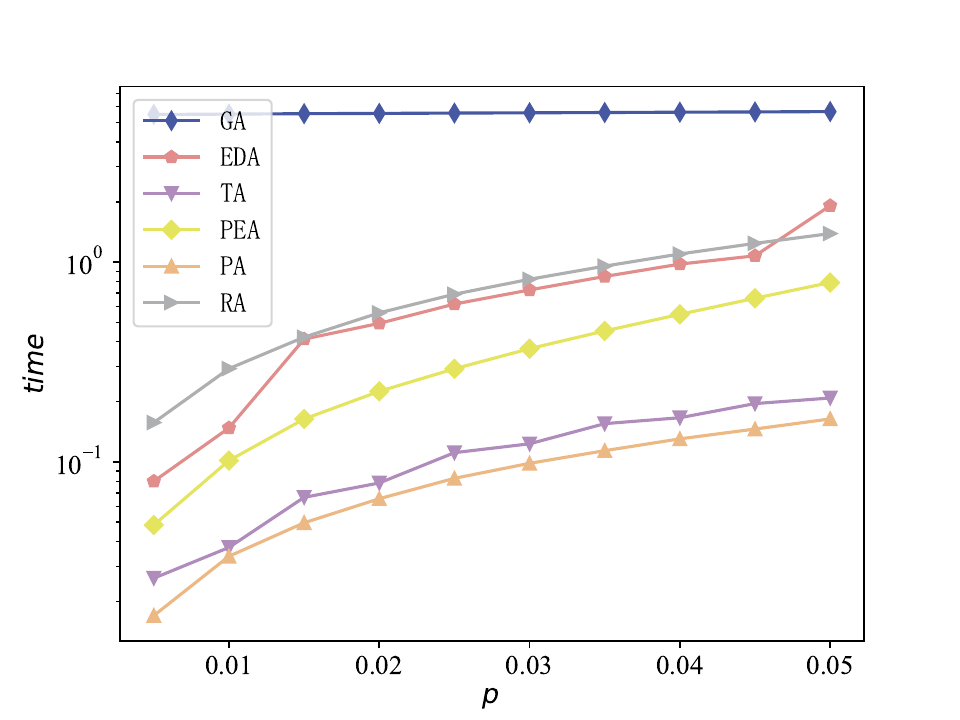}
        \end{minipage}
    }
  % \vspace{-0.3cm}
	% \\
	% \subfloat[Bus-685]{
 %        \begin{minipage}[b]{.3\linewidth}
 %            \centering
 %            \includegraphics[scale=0.3]{assort_time_image/Bus-685_time.pdf}
 %        \end{minipage}
 %    }
 %  % \hspace{2cm}
	% \subfloat[USAir97]{
 %        \begin{minipage}[b]{.3\linewidth}
 %            \centering
 %            \includegraphics[scale=0.3]{assort_time_image/USAir97_time.pdf}
 %        \end{minipage}
 %    }
 %  % \hspace{2cm}
 %    \subfloat[USAir500]{
 %        \begin{minipage}[b]{.3\linewidth}
 %            \centering
 %            \includegraphics[scale=0.3]{assort_time_image/USAir500_time.pdf}
 %        \end{minipage}	
 %    }
	\caption{The running time of five heuristics is analyzed as a function of the percentage $p$ of rewired edge pairs.}
	\label{fig:2} 
\end{figure*}
\subsubsection{The Comparison with Alternative Baselines}
We compare our proposed GA method and heuristic methods with the baseline methods described in Sec \ref{sec} on the real networks presented in Table \ref{tab1}, validating the effectiveness of our algorithm on real networks.

\begin{table*}[htbp]
\begin{center}
\caption{When the number of rewired edge pairs in the network is 5\% of the total number of edges, the GA method and our proposed heuristic methods are compared with baseline methods for rewiring the assortativity coefficient of three types of real networks. The text in red font corresponds to the highest assortativity coefficient among the six methods, while the text in blue font corresponds to the second highest assortativity coefficient.}
\label{tab4567}
\begin{tabular}{cccccccccc}
\hline
 & & & & & & & & & \\[-6pt]
Methods &AS-733-A&AS-733-B&AS-733-C&AS-733-D&AS-733-E&USPowerGrid&BCSPWR10&USAir97&USAir10
\\ \hline
GA   & \textcolor{red}{-0.214} & \textcolor{red}{-0.198}& \textcolor{red}{-0.191}&\textcolor{red}{-0.178}&\textcolor{red}{-0.172}&\textcolor{red}{0.556}&\textcolor{red}{0.502}&\textcolor{red}{-0.119}&\textcolor{red}{0.032}
\\ 
EDA   & -0.221 & -0.204& -0.196 &-0.182&-0.177&\textcolor{blue}{0.539}&\textcolor{blue}{-0.175}&\textcolor{blue}{-0.165}&\textcolor{blue}{-0.031}
\\ 
TA   & -0.221 & -0.204& -0.196 &-0.182&-0.177&0.464&0.403&\textcolor{blue}{-0.165}&-0.036
\\ 
PEA   & \textcolor{blue}{-0.218} & \textcolor{blue}{-0.201}& \textcolor{blue}{-0.194} &\textcolor{blue}{-0.180}&\textcolor{blue}{-0.175}&0.185&0.132&\textcolor{blue}{-0.165}&-0.043
\\ 
PA   & -0.224 & -0.207& -0.198 &-0.185&-0.180&0.073&0.032&-0.189&-0.083
\\ 
RA   & -0.223 & -0.206& -0.198 &-0.184&-0.178&0.069&0.02&-0.183&-0.073
\\ 
\hline
\end{tabular}
\end{center}
\end{table*} 

\begin{table*}[htbp]
\begin{center}
\caption{When the number of rewired edge pairs in the network is 5\% of the total number of edges, the GA method and our proposed heuristic methods are compared with baseline methods for rewiring the Spearman rank correlation coefficient of three types of real networks. The text in red font corresponds to the highest Spearman rank correlation coefficient among the six methods, while the text in blue font corresponds to the second highest Spearman rank correlation coefficient.}
\label{tab8888}
\begin{tabular}{cccccccccc}
\hline
 & & & & & & & & & \\[-6pt]
Methods &AS-733-A&AS-733-B&AS-733-C&AS-733-D&AS-733-E&USPowerGrid&BCSPWR10&USAir97&USAir10
\\ \hline
original & -0.504 & -0.481& -0.502 &-0.521&-0.050&-0.074&-0.144&-0.144
&-0.066
\\ 
GA   & \textcolor{red}{-0.227} & \textcolor{red}{-0.196}& \textcolor{red}{-0.212}&\textcolor{red}{-0.211}&\textcolor{red}{-0.230}&\textcolor{red}{0.245}&\textcolor{red}{0.258}&\textcolor{red}{0.030}&\textcolor{red}{0.156}
\\ 
EDA   & \textcolor{blue}{-0.309} & \textcolor{blue}{-0.289}& \textcolor{blue}{-0.312} &\textcolor{blue}{-0.324}&\textcolor{blue}{-0.351}
&\textcolor{blue}{0.223}&\textcolor{blue}{0.240}&\textcolor{blue}{-0.052}&0.054
\\ 
TA   & -0.310 & \textcolor{blue}{-0.289}& \textcolor{blue}{-0.312} &-0.326&0.352&0.100&0.098&\textcolor{blue}{-0.052}&\textcolor{blue}{0.059}
\\ 
PEA   & -0.368 & -0.347& -0.366 &-0.367&-0.384&0.112&0.094&-0.070&0.028
\\ 
PA   & -0.428 & -0.407& -0.425 &-0.426&-0.445&0.042&0.027&-0.110&-0.024
\\ 
RA   & -0.407 & -0.387& -0.405 &-0.407&-0.424 &0.039&0.015&-0.098&-0.009
\\ 
\hline
\end{tabular}
\end{center}
\end{table*} 

To ensure the validity of the experiments, we repeated the experiments 50 times on real networks for methods with uncertain results, such as RA, and averaged the results. Table \ref{tab4567} displays the assortativity coefficients of the real networks after rewiring by our GA method and heuristic methods, compared to baseline methods, when the rewiring budget is 5\% of the total number of edges in the network. The GA method consistently achieves the best results across all three types of networks, while our proposed heuristic methods EDA, TA, and PEA also outperform the baseline methods on all networks. We observe that the performance of the three heuristic methods varies across different types of networks. In the routing network, the performance of PEA is second only to the GA method. In the power network, EDA and TA perform well, especially EDA, which closely matches the increase in network assortativity coefficients achieved by the GA method.  In the flight network, our three heuristic methods show similar effectiveness. Notably, EDA and TA demonstrate similar effects across all three types of networks. This suggests that although our EDA and TA methods employ different strategies for rewiring edge pairs, they tend to select similar edge pairs for rewiring. One possible explanation is that the TA method prioritizes rewiring edge pairs involving high-degree nodes, similar to the edge pairs with large degree differences targeted by the EDA method. This phenomenon is particularly prominent in disassortative real networks. 

Another noteworthy phenomenon emerges when considering neutral networks:  for neutral networks, our methods exhibit a significant improvement in the network assortativity coefficient. For instance, in the power network, the GA method increases the assortativity coefficients of USPowerGrid and BCSPWR10 by 0.553 and 0.507, respectively. This transformation effectively changes them from neutral networks into strongly assortative networks. In contrast, for disassortative scale-free networks, even the improvement in the assortativity coefficient achieved by the GA method is limited. For example, in AS-733-A and AS-733-E, the GA method increases their assortativity coefficients by only 0.015 and 0.010, respectively. The reason behind this phenomenon lies in the influence of network degree distribution on the value of the assortativity coefficient.  Scale-free networks with degree exponent $\gamma < 3$ tend to exhibit structural disassortativity \cite{boguna2004cut}(e.g., $\gamma_{AS-733-A}=2.20$, $\gamma_{AS-733-E}=2.11$), indicating the presence of multiple edges between high-degree nodes. However, due to the limitation of being a simple network with only one edge between nodes, the network tends to be disassortative. Additionally, the range within which the network's assortativity coefficient can vary is relatively small. Although rewiring effectively changes the network's structure, , these changes may not be prominently reflected in the assortativity coefficient.

We can evaluate the degree correlation of networks demonstrating structural disassortativity using the Spearman rank correlation coefficient \cite{zhang2016measuring}. In Sec. \ref{rocm}, the calculation of the Spearman rank correlation coefficient for centrality measures is described to assess their robustness. Here, we calculate the Spearman rank correlation coefficient based on node degrees to measure the degree correlation of the network. The Spearman rank correlation coefficient utilizes the rankings of node degrees instead of their actual degrees, thereby reducing the influence of degree distribution on the assortativity coefficient. It is evident from Table \ref{tab8888} that the Spearman rank correlation coefficient effectively captures the degree of change in degree correlation in disassortative scale-free networks. For example, in AS-733-A, the GA method increases the network's Spearman rank correlation coefficient by 0.227. Furthermore, while PEA demonstrates superior performance to EDA and TA in terms of the assortativity coefficient, EDA and TA outperform PEA when considering the Spearman rank correlation coefficient in certain networks. This indicates that the Spearman rank correlation coefficient, which considers the rankings of node degrees, may not always align well with the assortativity coefficient.

Figure \ref{fig:1} depicts the assortativity coefficient variations of the network under different methods for rewiring budgets ranging from 0.5\% to 5\% of the number of network edges. The trends observed in the routing network are similar, thus, we present a subset of networks here. We can clearly see that the GA method yields the best results. Across all routing networks, different methods exhibit similar effects, with GA being the most effective, followed by PEA, while EDA and TA show comparable performance, and PA and RA methods are the least effective. Similar observations can be made for the power networks, although PEA and TA significantly outperform EDA. In the power networks, our heuristic methods, PEA and TA, show improvements in assortativity coefficients that are very close to those achieved by the GA method, especially the EDA method. In flight network, the performance of the three methods we proposed is similar, with only slight variations. Specifically, in USAir97, PEA is slightly better than EDA and TA, while in USAir10, EDA and TA are slightly better than PEA.

Next, we conduct an analysis of the time efficiency of our GA method and the heuristic methods in comparison to baseline methods. The Figure \ref{fig:2} illustrates the runtime of different methods across three types of networks as the number of rewirings ranges from 0.05\% to 5\% of the total number of edges in the network. We observe that the time efficiency of the GA method is notably lower, differing by several orders of magnitude from the other methods. Additionally, as the network scale increases, the time cost of the GA method sharply rises. It is noteworthy that our GA only performs one initial sorting of the $value$ for all possible edge pairs with positive $value$, so the number of rewirings typically does not significantly affect its runtime. The runtime for the EDA, TA, and PEA methods is similar to that of baseline methods, and in some networks, it even outperforms baseline methods. Therefore, in conjunction with the preceding experiments, our proposed heuristic methods demonstrate a clear advantage over baseline methods and effectively increase the assortativity coefficient of networks. This suggests that when the network scale is large and GA is impractical, EDA, TA, and PEA can be flexibly employed based on the network type. For example, in power networks, EDA and TA are favored, whereas PEA is better suited for router networks.

\begin{figure*}[htbp]%调节图片位置，h：浮动；t：顶部；b:底部；p：当前位置]
	\centering
	\subfloat[AS-733-A]{
        \begin{minipage}[b]{.3\linewidth}
            \centering
            \includegraphics[scale=0.3]{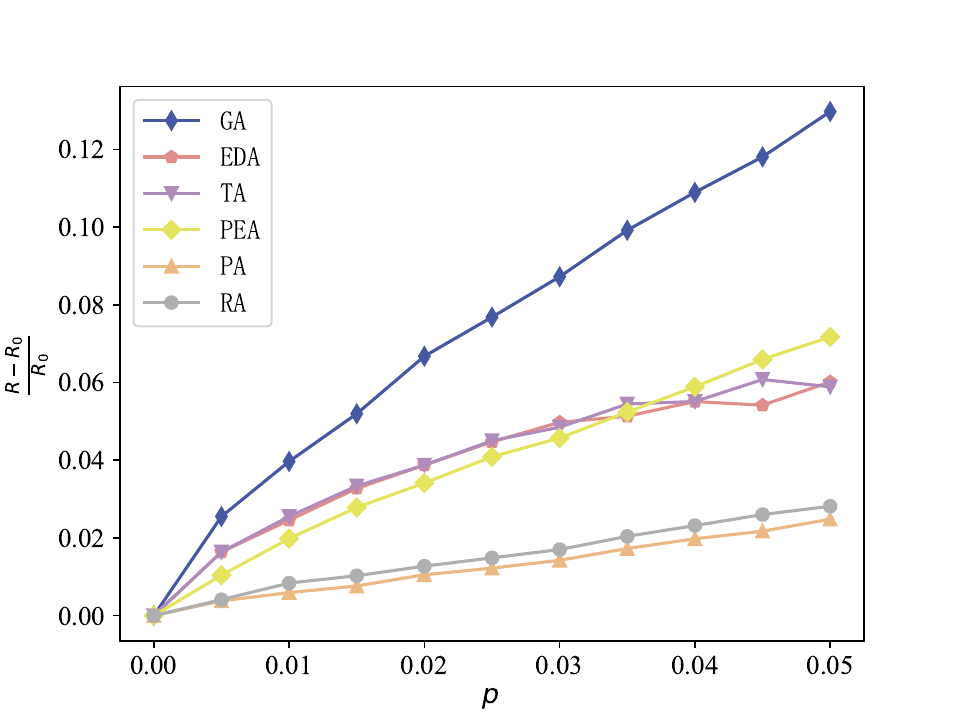}
        \end{minipage}
    }
  % \hspace{2cm}
    \subfloat[USPowerGrid]{
        \begin{minipage}[b]{.3\linewidth}
            \centering
            \includegraphics[scale=0.3]{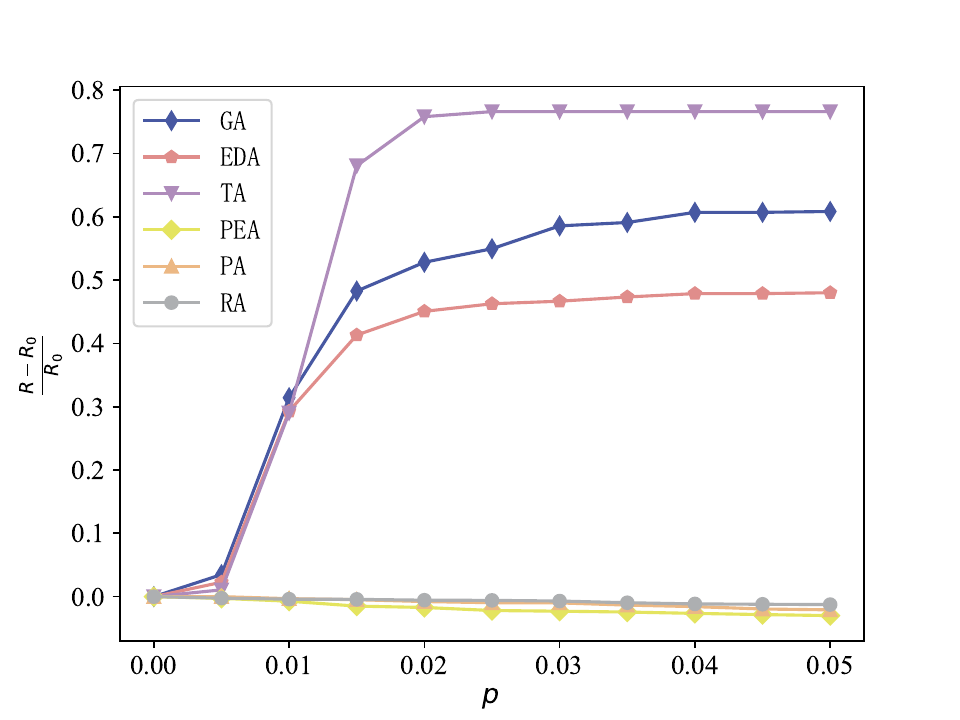}
        \end{minipage}
    }
      % \hspace{2cm}
	\subfloat[USAir97]{
        \begin{minipage}[b]{.3\linewidth}
            \centering
            \includegraphics[scale=0.3]{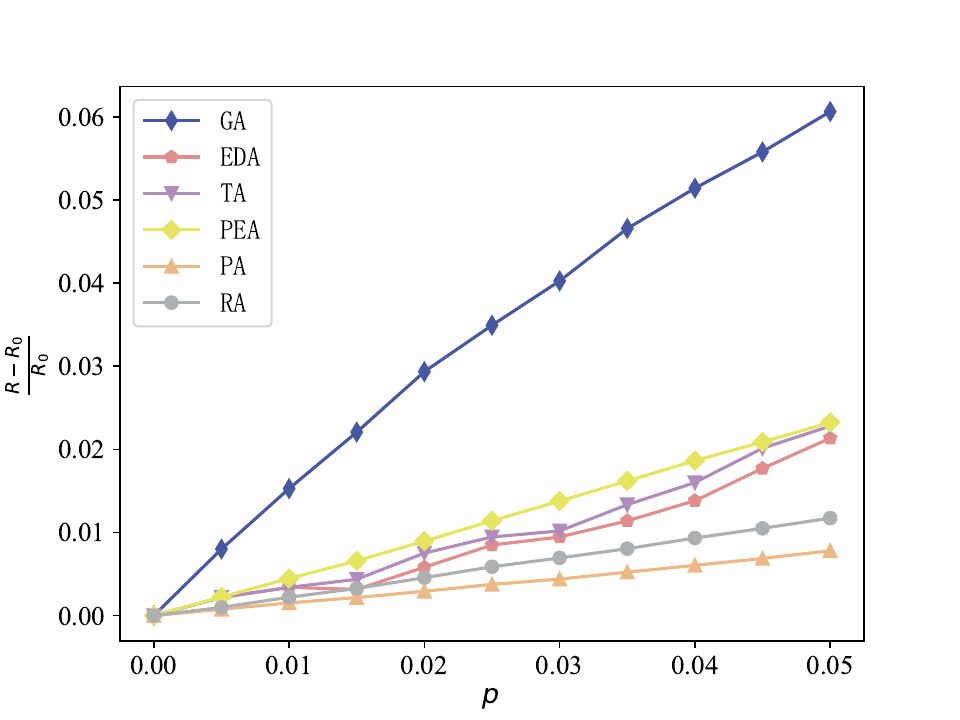}
        \end{minipage}
    }
  \vspace{-0.3cm}
	\\
   % \hspace{-0.5cm}
	\subfloat[AS-733-E]{
        \begin{minipage}[b]{.3\linewidth}
            \centering
            \includegraphics[scale=0.3]{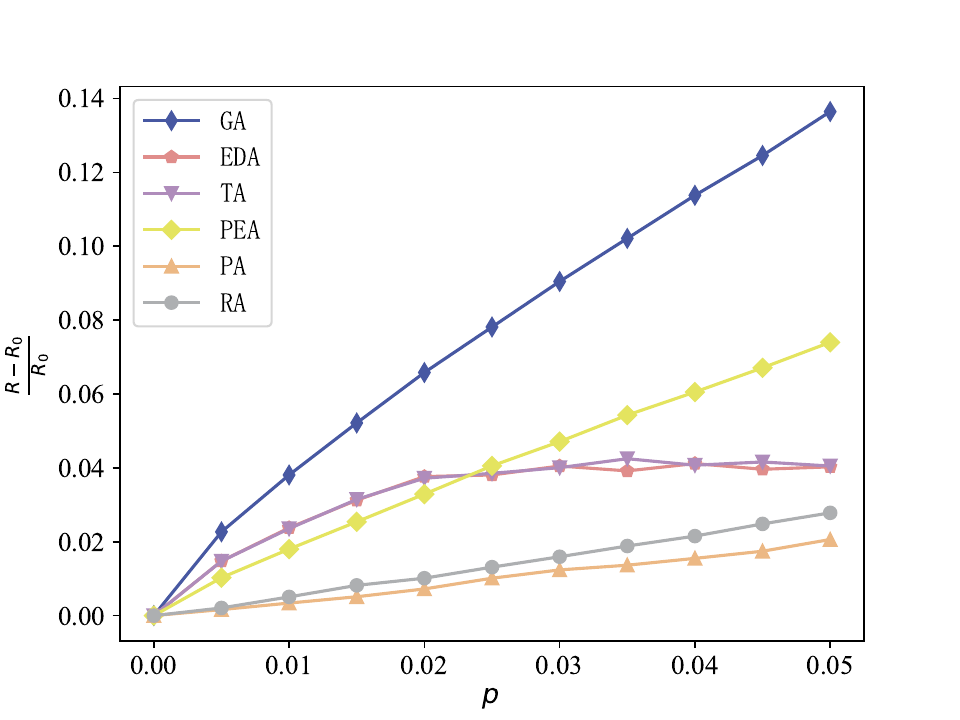}
        \end{minipage}
    }
	\subfloat[BCSPWR10]{
        \begin{minipage}[b]{.3\linewidth}
            \centering
            \includegraphics[scale=0.3]{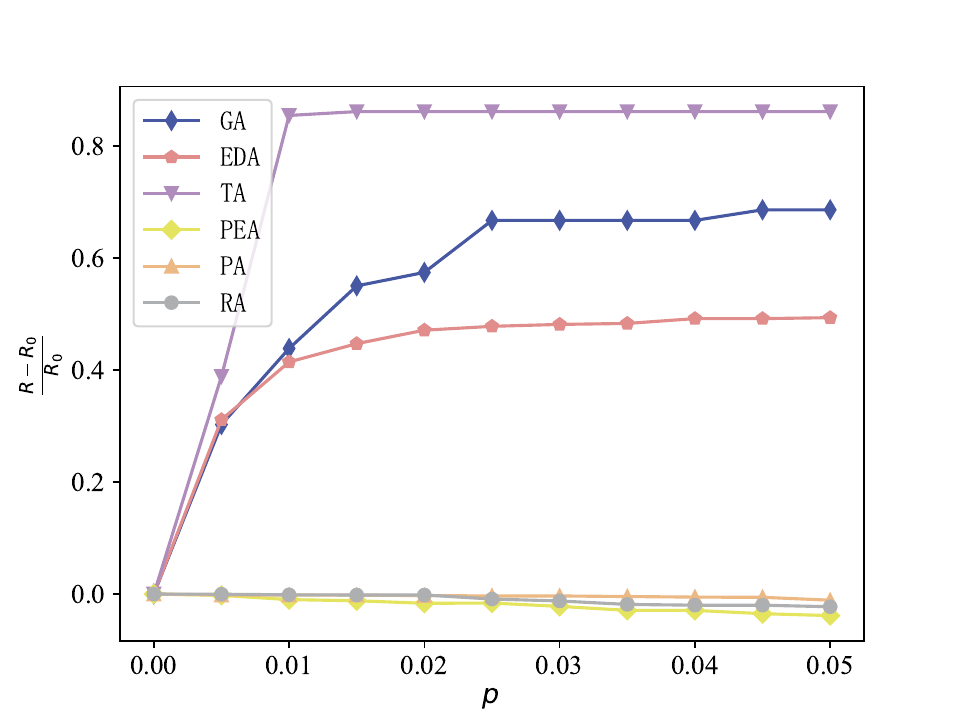}
        \end{minipage}
    }
  % \hspace{2cm}
    \subfloat[USAir10]{
        \begin{minipage}[b]{.3\linewidth}
            \centering
            \includegraphics[scale=0.3]{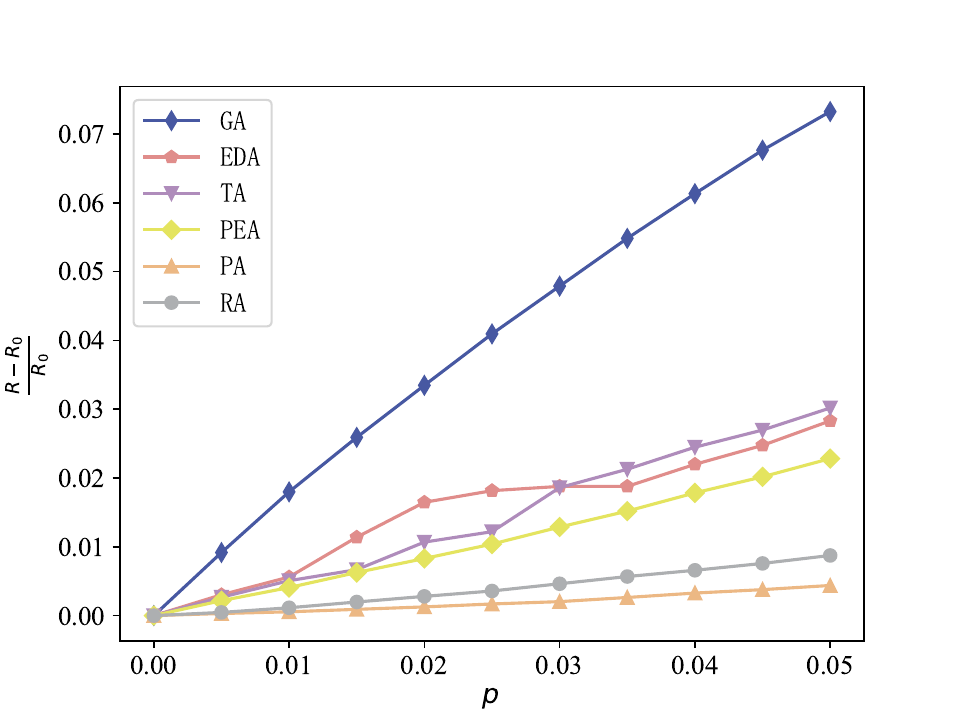}
        \end{minipage}	
    }
	\caption{The spectral radius of five heuristics is analyzed as a function of the percentage $p$ of rewired edge pairs.}
	\label{fig:9} 
\end{figure*}
% \begin{figure*}[htbp]%调节图片位置，h：浮动；t：顶部；b:底部；p：当前位置]
% 	\centering
% 	\subfloat[AS-733-A]{
%         \begin{minipage}[b]{.3\linewidth}
%             \centering
%             \includegraphics[scale=0.3]{sg_image/AS-733-A-sg.pdf}
%         \end{minipage}
%     }
%   % \hspace{-0.5cm}
% 	\subfloat[AS-733-E]{
%         \begin{minipage}[b]{.3\linewidth}
%             \centering
%             \includegraphics[scale=0.3]{sg_image/AS-733-E-sg.pdf}
%         \end{minipage}
%     }
%   % \hspace{2cm}
%     \subfloat[USPowerGrid]{
%         \begin{minipage}[b]{.3\linewidth}
%             \centering
%             \includegraphics[scale=0.3]{sg_image/USPowerGrid-sg.pdf}
%         \end{minipage}
%     }
%   \vspace{-0.3cm}
% 	\\
% 	\subfloat[BCSPWR10]{
%         \begin{minipage}[b]{.3\linewidth}
%             \centering
%             \includegraphics[scale=0.3]{sg_image/BCSPWR10-sg.pdf}
%         \end{minipage}
%     }
%   % \hspace{2cm}
% 	\subfloat[USAir97]{
%         \begin{minipage}[b]{.3\linewidth}
%             \centering
%             \includegraphics[scale=0.3]{sg_image/USAir97-sg.pdf}
%         \end{minipage}
%     }
%   % \hspace{2cm}
%     \subfloat[USAir10]{
%         \begin{minipage}[b]{.3\linewidth}
%             \centering
%             \includegraphics[scale=0.3]{sg_image/USAir10-sg.pdf}
%         \end{minipage}	
%     }
% 	\caption{The spectral gap of five heuristics is analyzed as a function of the percentage $p$ of rewired edge pairs.}
% 	\label{fig:11} 
% \end{figure*}
\begin{figure*}[htbp]%调节图片位置，h：浮动；t：顶部；b:底部；p：当前位置]
	\centering
	\subfloat[AS-733-A]{
        \begin{minipage}[b]{.3\linewidth}
            \centering
            \includegraphics[scale=0.245]{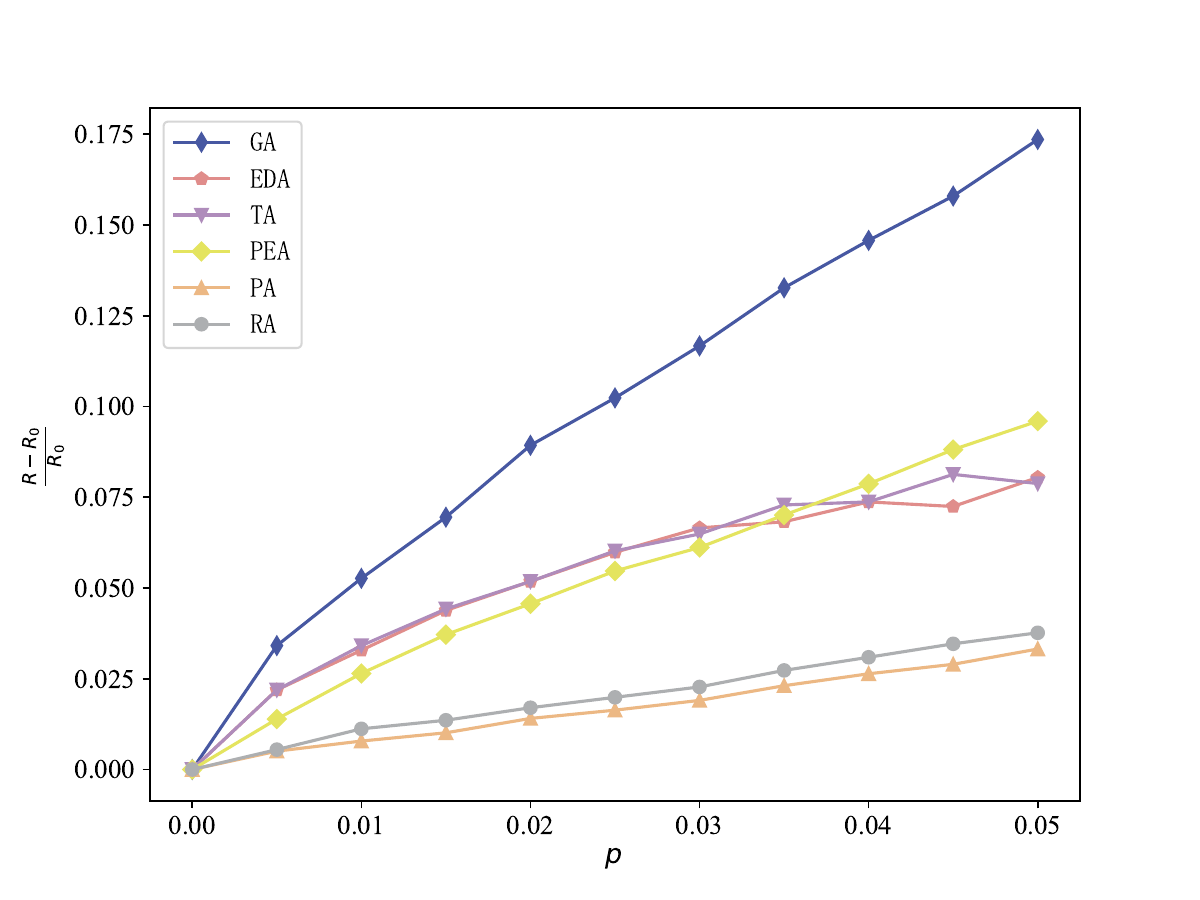}
        \end{minipage}
        }
         % \hspace{0.1cm}
        \subfloat[USPowerGrid]{
            \begin{minipage}[b]{.3\linewidth}
                \centering
                \includegraphics[scale=0.3]{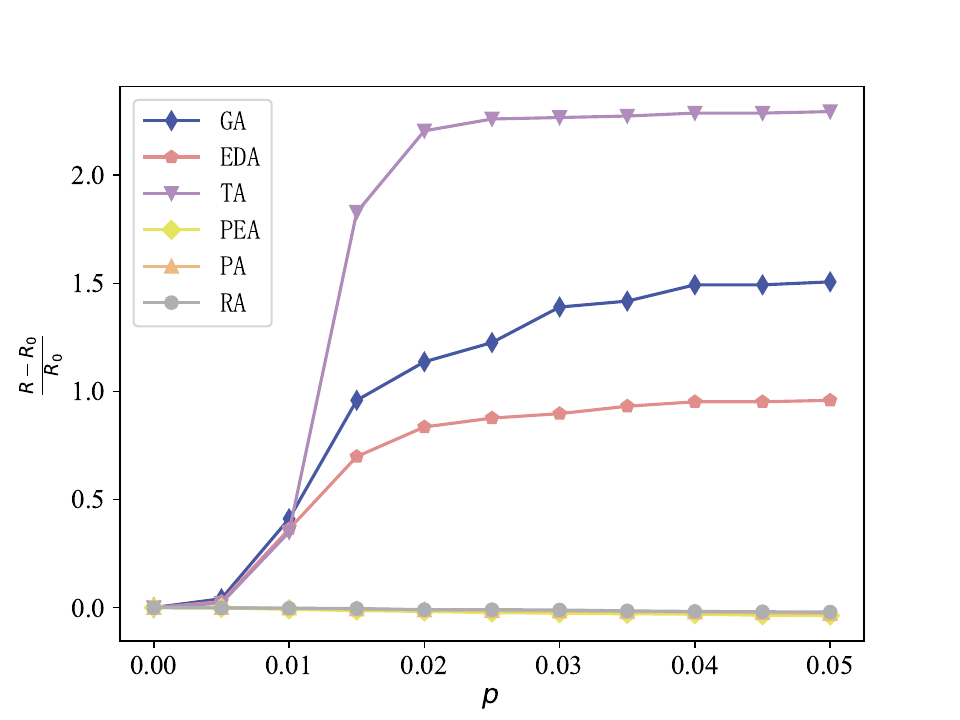}
            \end{minipage}
        }
          % \hspace{-0.1cm}
        \subfloat[USAir97]{
            \begin{minipage}[b]{.3\linewidth}
                \centering
                \includegraphics[scale=0.3]{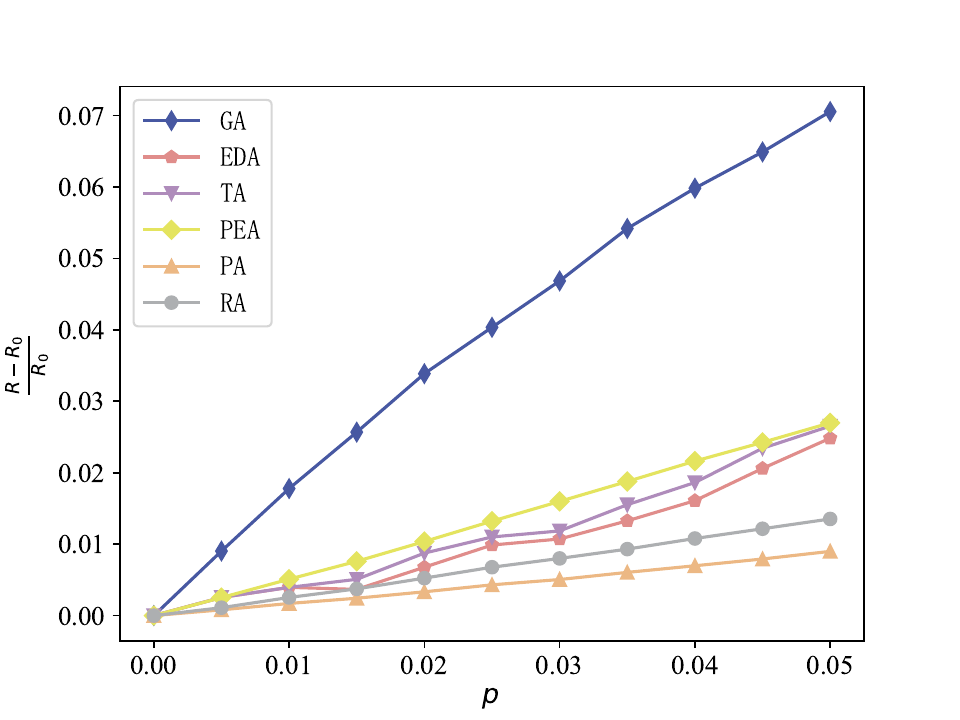}
            \end{minipage}
        }
  \vspace{-0.3cm}
  \\

  % \hspace{2cm}

    \subfloat[AS-733-E]{
        \begin{minipage}[b]{.3\linewidth}
            \centering
            \includegraphics[scale=0.245]{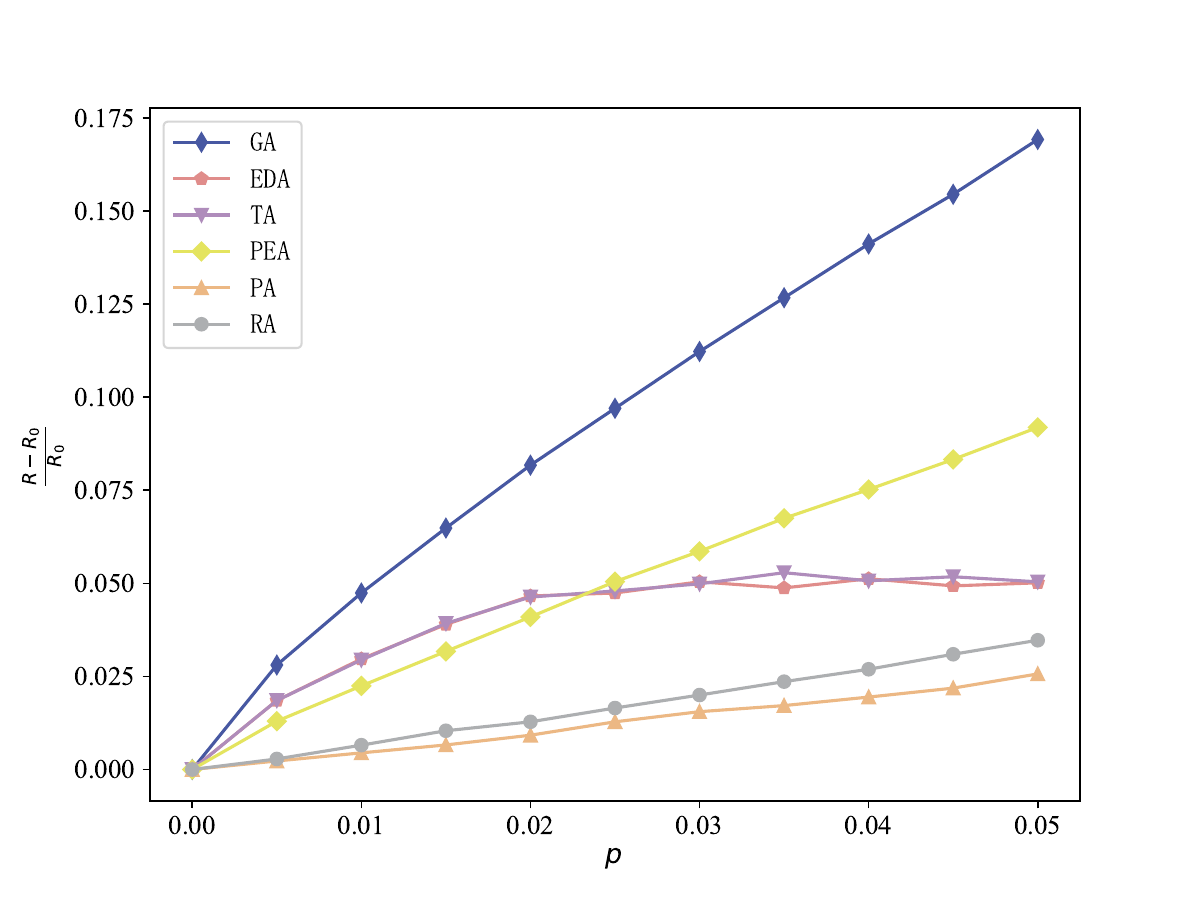}
        \end{minipage}
    }
    \subfloat[BCSPWR10]{
        \begin{minipage}[b]{.3\linewidth}
            \centering
            \includegraphics[scale=0.3]{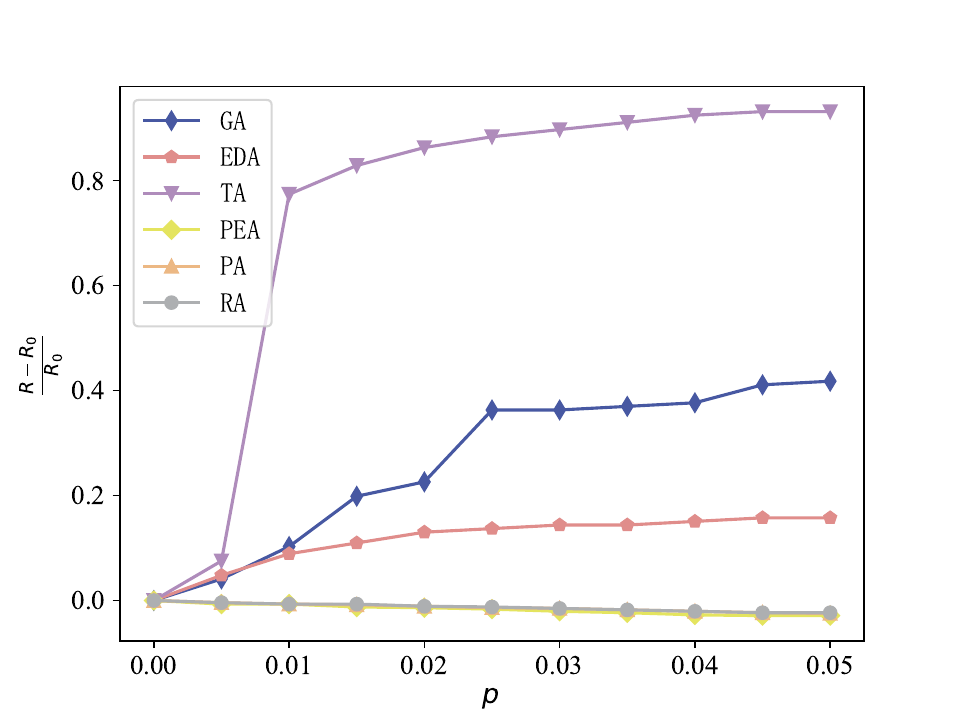}
        \end{minipage}
    }
  % \hspace{2cm}
    \subfloat[USAir10]{
        \begin{minipage}[b]{.3\linewidth}
            \centering
            \includegraphics[scale=0.3]{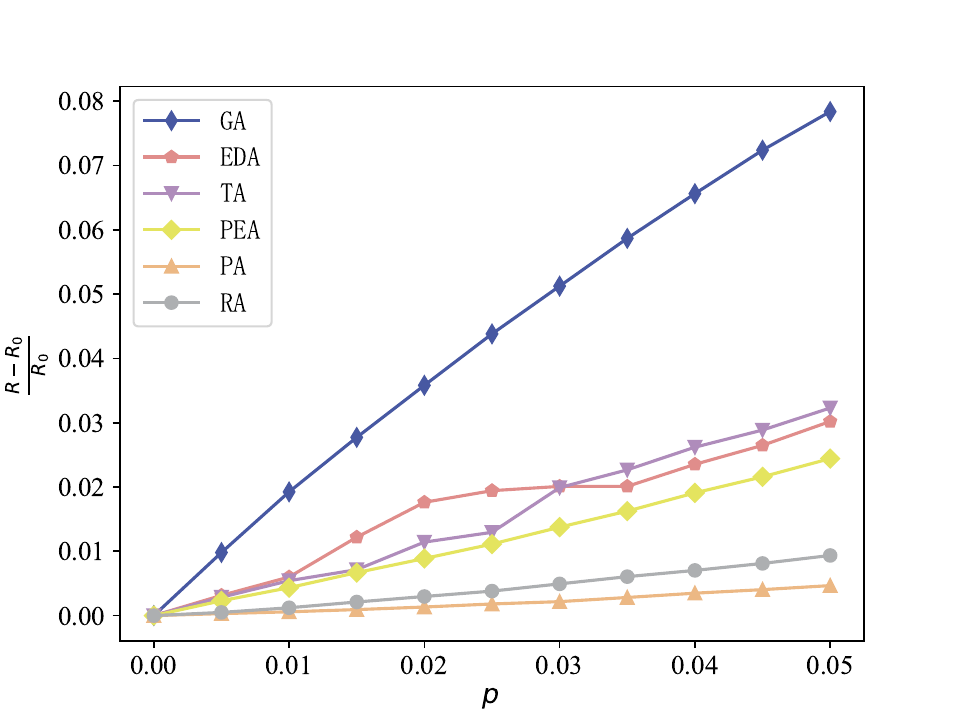}
        \end{minipage}	
    }
	\caption{The natural connectivity  of five heuristics is analyzed as a function of the percentage $p$ of rewired edge pairs.}
	\label{fig:12} 
\end{figure*}
\subsection{The Analysis of Network Robustness}
In this section, we analyze the impact of the GA method and the heuristic methods on network robustness by selecting several representative measures, as described in Section \ref{nr}. We compare the changes in these robustness measures before and after executing the rewiring methods, considering a rewiring budget ranging from 0.5\% to 5\% of the number of network edges.
%is 5\% of the number of network edges.
% for rewiring budgets ranging from 0.5% to 5% of the number of network edges

% According to the definitions of the three spectral robustness metrics, it can be observed that they are all directly related to the largest eigenvalue of the network's adjacency matrix. Increasing the network's assortativity coefficient typically leads to an increase in the largest eigenvalue of the network, thereby enhancing the robustness metrics associated with the largest eigenvalue. As shown in Figure \ref{fig:9}, with the increase in rewiring iterations, all five rewiring strategies are capable of increasing the spectral radius of both the AS router network and the flight network. Figures \ref{fig:11} and \ref{fig:12} also indicate that as the number of rewiring iterations increases, the spectral gap and natural connectivity of the networks are both enhanced, with GA exhibiting the best performance. However, the correspondence between assortativity coefficient and largest eigenvalue is not straightforward; For example, in the power networks, GA is the most effective method among all for increasing network assortativity, while TA is the most effective method for increasing network spectral radius. 

Figure \ref{fig:9} illustrates the variation of the spectral radius under different rewiring methods. We use $\frac{R-R_0}{R_0}$ as the vertical axis to represent the corresponding change rate in robustness metrics. Similarly, Figures \ref{fig:12} shows the changes in natural connectivity under different rewiring methods. 

According to the definitions of the two spectral robustness metrics, it can be observed that they are all directly related to the largest eigenvalue of the network's adjacency matrix. Increasing the network's assortativity coefficient typically leads to an increase in the largest eigenvalue of the network, thereby enhancing the robustness metrics associated with the largest eigenvalue. Figures \ref{fig:9} and \ref{fig:12} demonstrate that the variation trend of the spectral radius and the natural connectivity under different rewiring methods in routing and flight networks is similar to that of the assortativity coefficient. Specifically, the rewiring methods that are more effective in increasing the network's assortativity coefficient also tend to effectively increase the network's spectral radius and natural connectivity in these two types of networks. While the relationship between the assortativity coefficient and the largest eigenvalue is not straightforward, particularly in power networks, some interesting observations emerge.  For instance, in power networks, the GA method proves most effective in increasing the network assortativity, whereas TA emerges as the most effective method for enhancing the network's spectral radius. Moreover, EDA, TA, and GA methods initially lead to a rapid increase in the network's spectral radius with an uptick in rewiring frequency, stabilizing once the rewiring frequency surpasses 2.5\% of the total number of edges, with no further increase observed with additional rewiring. Additionally, despite RA, PA, and PEA's capacity to augment the network's assortativity coefficient, they do not contribute to improvements in the network's spectral radius and natural connectivity.

Observing Figures \ref{fig:9} and \ref{fig:12} reveals an interesting phenomenon: the variations in the natural connectivity of different network types under different rewiring methods resemble those of their spectral radius. One possible explanation is that natural connectivity represents the weighted average of all eigenvalues of the network adjacency matrix, with the maximum eigenvalue being predominant, thereby resulting in similar variations in spectral radius and natural connectivity.

 % For instance, in the power grid network, TA proves to be the most effective, stabilizing when the rewiring frequency exceeds 1.5\%.

Furthermore, we noted that the stability of the two robustness metrics varies across networks of different types. For example, in the AS router network and the flight network, when the rewiring ratio is 5\%, the increase in the spectral radius is 12\% and 14\% in the AS router network, and 6.7\% and 17.9\% in the flight network, respectively. However, in the power network, the increase in the spectral radius reaches as high as 78\% and 86\%. Similar phenomena are also observed in natural connectivity.

Overall, GA effectively improves the spectral robustness metrics of the three types of networks, with particularly notable performance in the router network and flight network compared to other rewiring strategies. Our three heuristic methods perform well in both routing and flight networks, with TA and EDA also proving effective for the power network. Notably, in the power network, TA outperforms GA.
It is worth noting that our rewiring strategy does not require the calculation of network robustness metrics at each rewiring step. Even spectral-based robustness metrics are computationally expensive, especially for large-scale networks. Therefore, our rewiring strategy demonstrates significant time efficiency.

% The changes in the robustness measures of the real networks from Table \ref{tab1} are presented in Table \ref{tab3}. Our GA algorithm effectively improves the robustness measures, such as spectral radius, spectral gap, and natural connectivity, while efficiently altering the network assortativity. There are already some methods designed to optimize network robustness, but many of these methods primarily focus on altering network structure to enhance individual robustness measures. Through experiments, we have validated that our GA method effectively improves multiple network robustness measures by increasing the network assortativity. This contributes to optimizing network structure and enhancing overall network robustness. 

\subsection{Robustness of centrality measures}

\begin{figure*}[htbp]%调节图片位置，h：浮动；t：顶部；b:底部；p：当前位置]
	\centering
	\subfloat[AS-733-A]{
        \begin{minipage}[b]{.3\linewidth}
            \centering
            \includegraphics[scale=0.3]{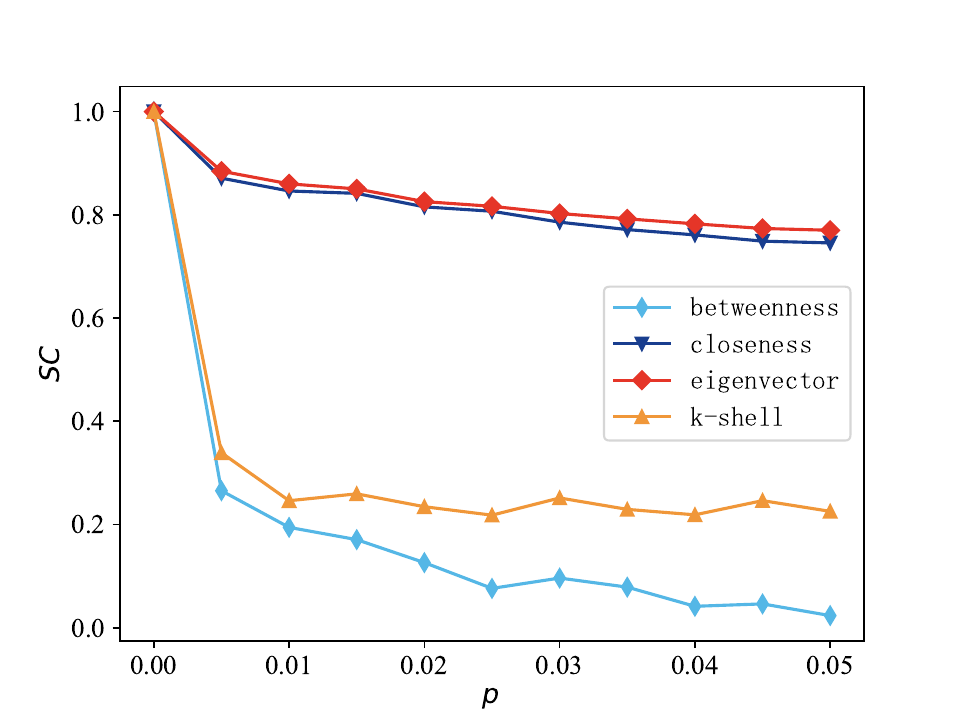}
        \end{minipage}
    }
    \subfloat[USPowerGrid]{
        \begin{minipage}[b]{.3\linewidth}
            \centering
            \includegraphics[scale=0.3]{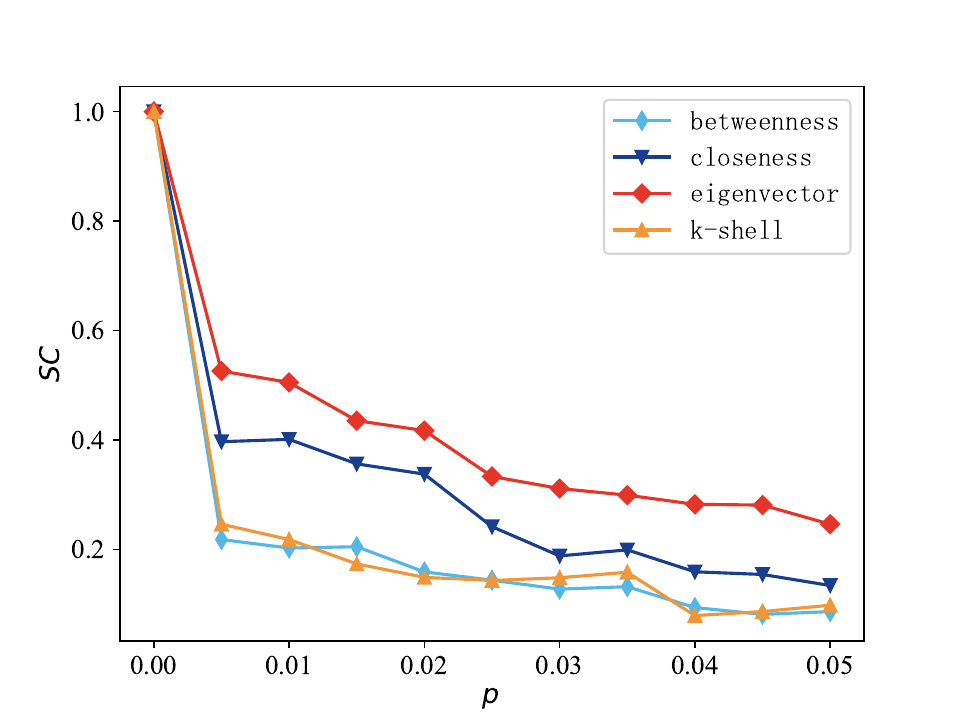}
        \end{minipage}
    }
    	\subfloat[USAir97]{
        \begin{minipage}[b]{.3\linewidth}
            \centering
            \includegraphics[scale=0.3]{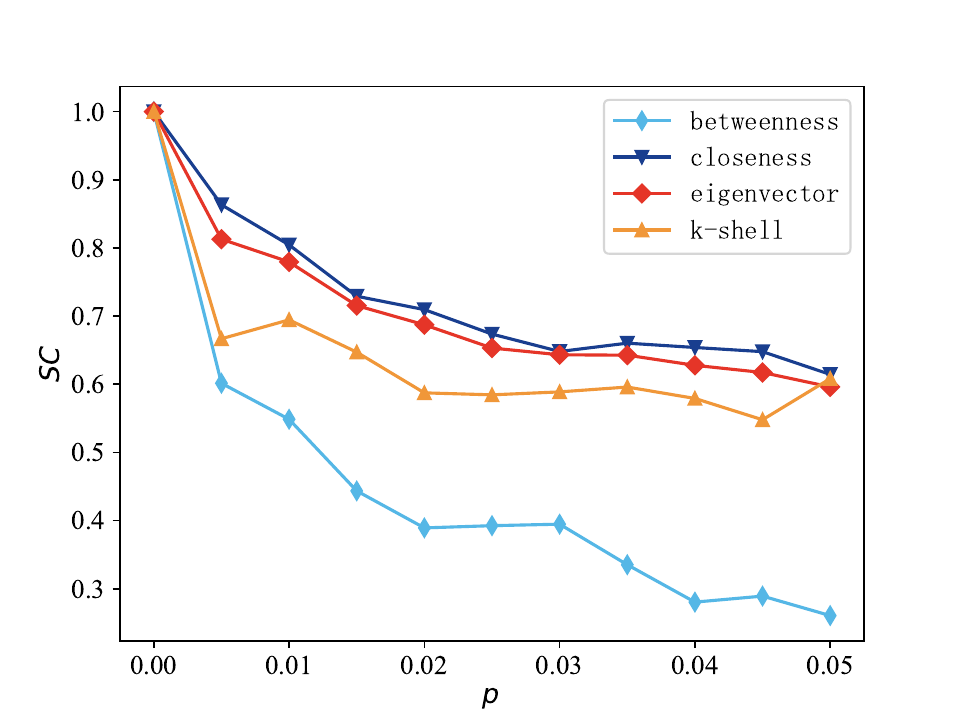}
        \end{minipage}
    }
  \vspace{-0.3cm}
	\\
 	\subfloat[AS-733-E]{
        \begin{minipage}[b]{.3\linewidth}
            \centering
            \includegraphics[scale=0.3]{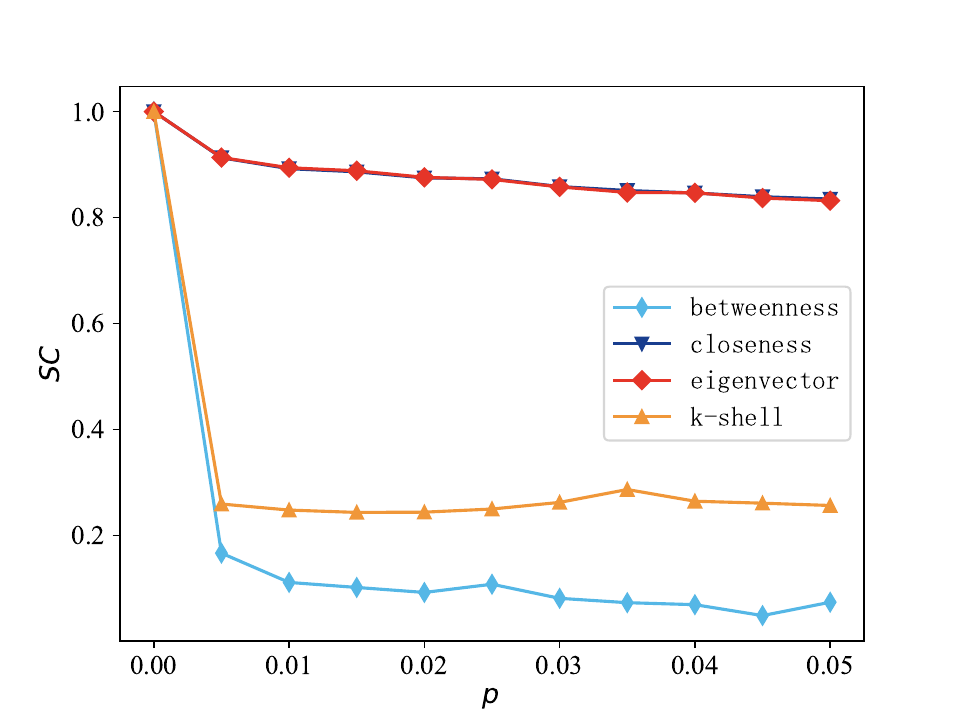}
        \end{minipage}
    }
	\subfloat[BCSPWR10]{
        \begin{minipage}[b]{.3\linewidth}
            \centering
            \includegraphics[scale=0.3]{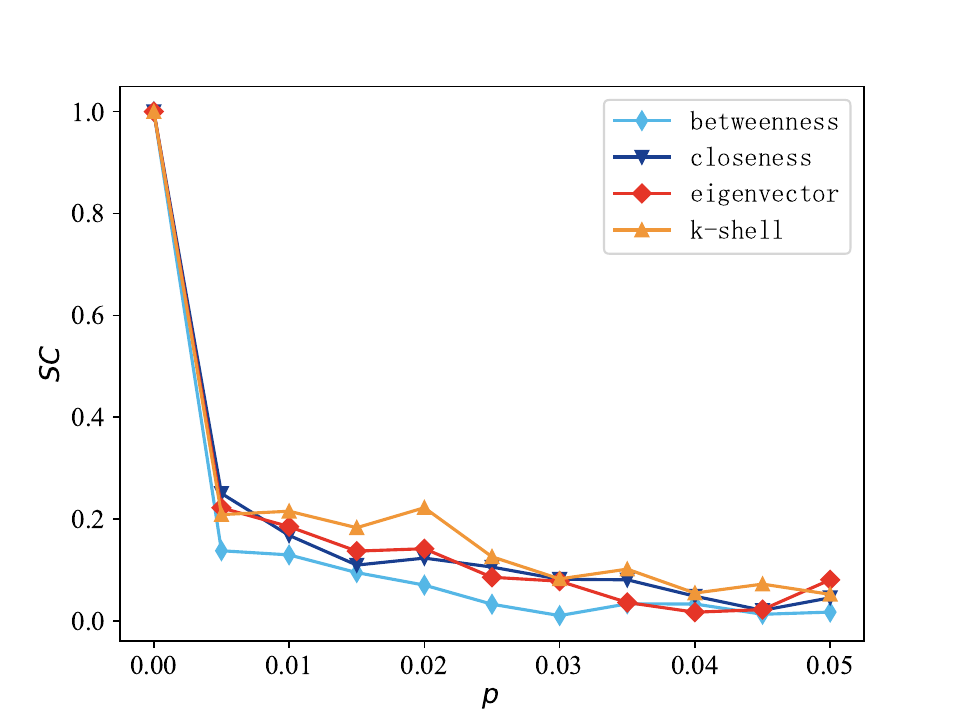}
        \end{minipage}
    }
    \subfloat[USAir10]{
        \begin{minipage}[b]{.3\linewidth}
            \centering
            \includegraphics[scale=0.3]{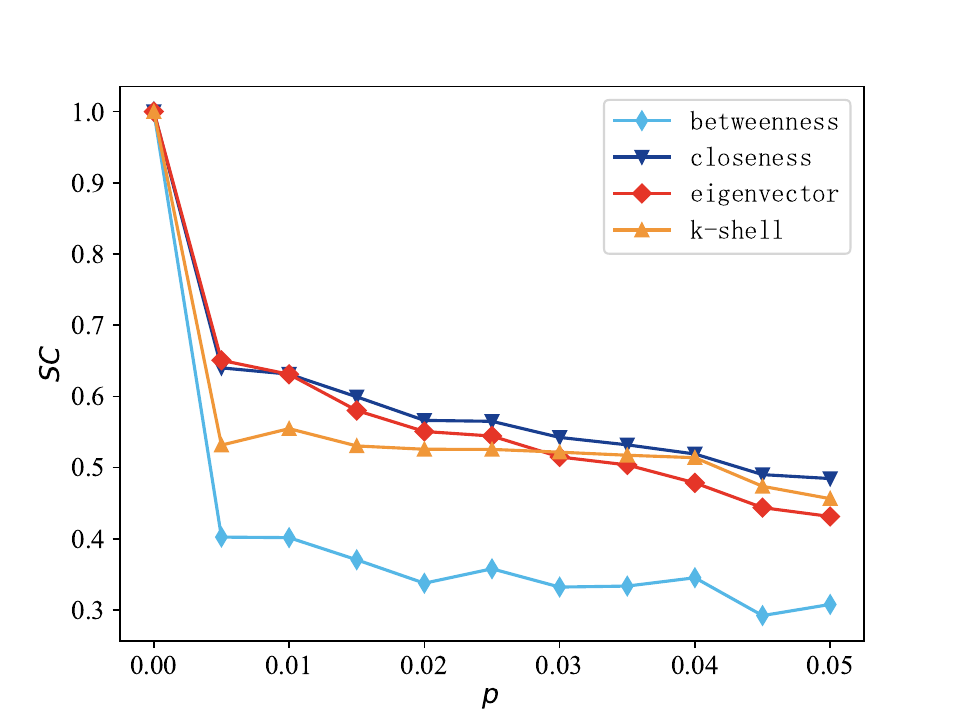}
        \end{minipage}	
    }
	\caption{The influence of rewiring edge pairs using the GA method on the Spearman rank correlation coefficient $SC$ between the true measure $C_T$ and manipulated measure $C_M$, with rewiring frequencies ranging from 0.5\% to 5\% of the total number of edges in the network.}
	\label{fig:2211} 
\end{figure*}

\begin{figure*}[htbp]%调节图片位置，h：浮动；t：顶部；b:底部；p：当前位置]
	\centering
	\subfloat[AS-733-A]{
        \begin{minipage}[b]{.3\linewidth}
            \centering
            \includegraphics[scale=0.3]{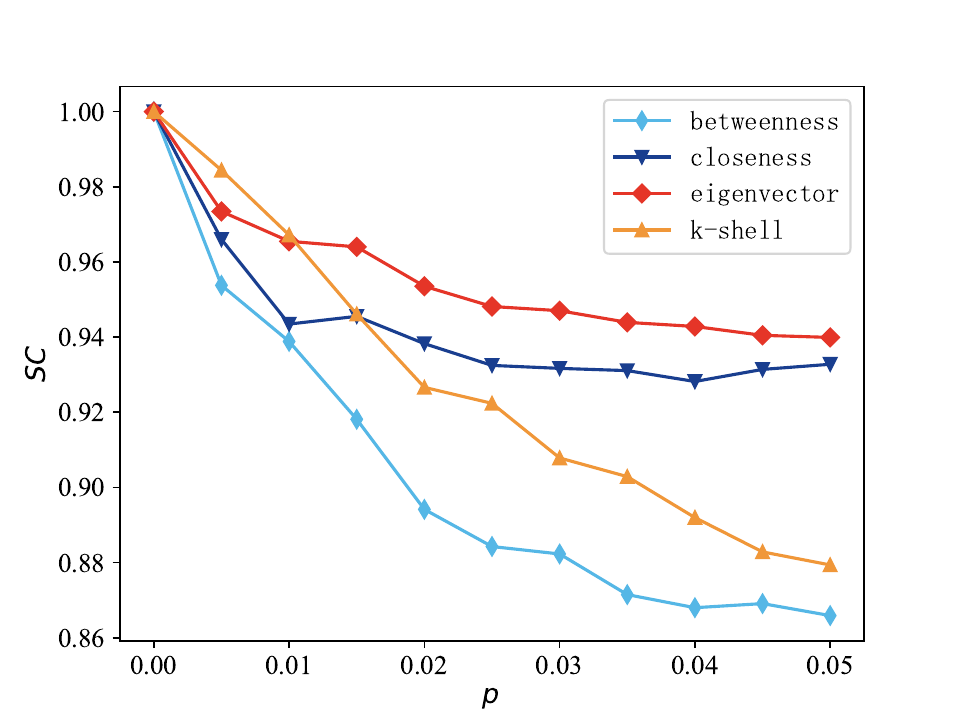}
        \end{minipage}
    }
  % \hspace{2cm}
    \subfloat[USPowerGrid]{
        \begin{minipage}[b]{.3\linewidth}
            \centering
            \includegraphics[scale=0.3]{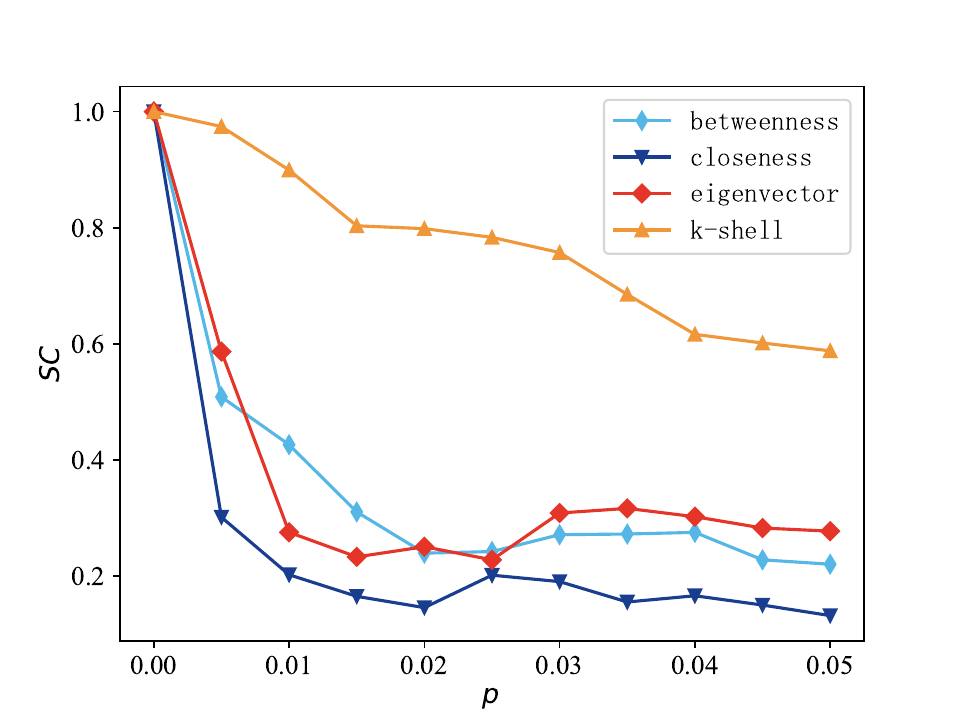}
        \end{minipage}
    }
	\subfloat[USAir97]{
        \begin{minipage}[b]{.3\linewidth}
            \centering
            \includegraphics[scale=0.3]{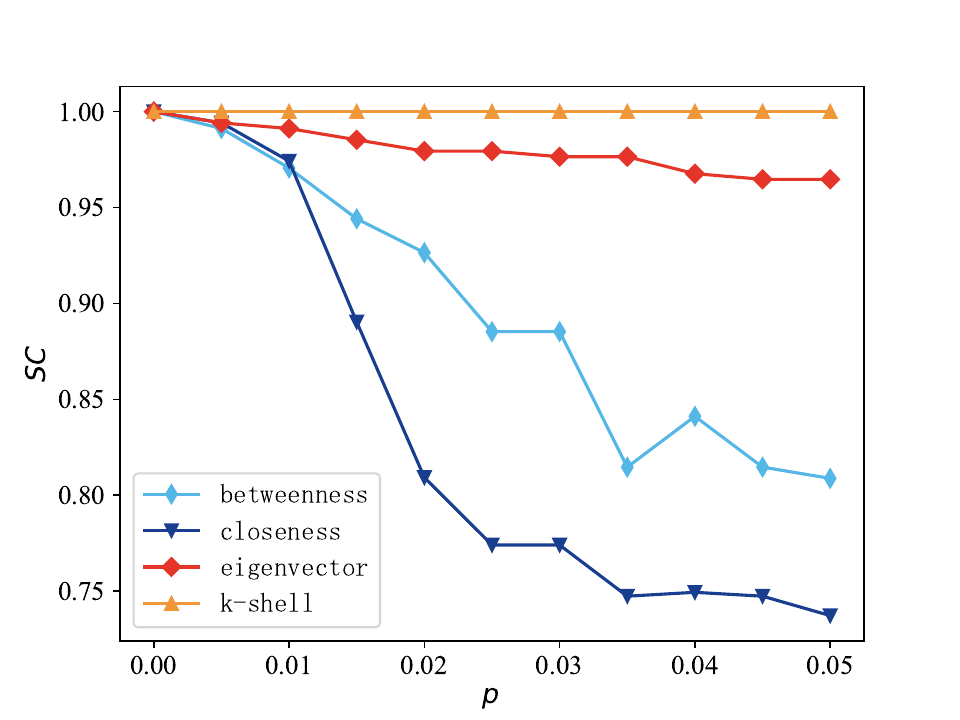}
        \end{minipage}
    }
  \vspace{-0.3cm}
	\\
	\subfloat[AS-733-E]{
        \begin{minipage}[b]{.3\linewidth}
            \centering
            \includegraphics[scale=0.3]{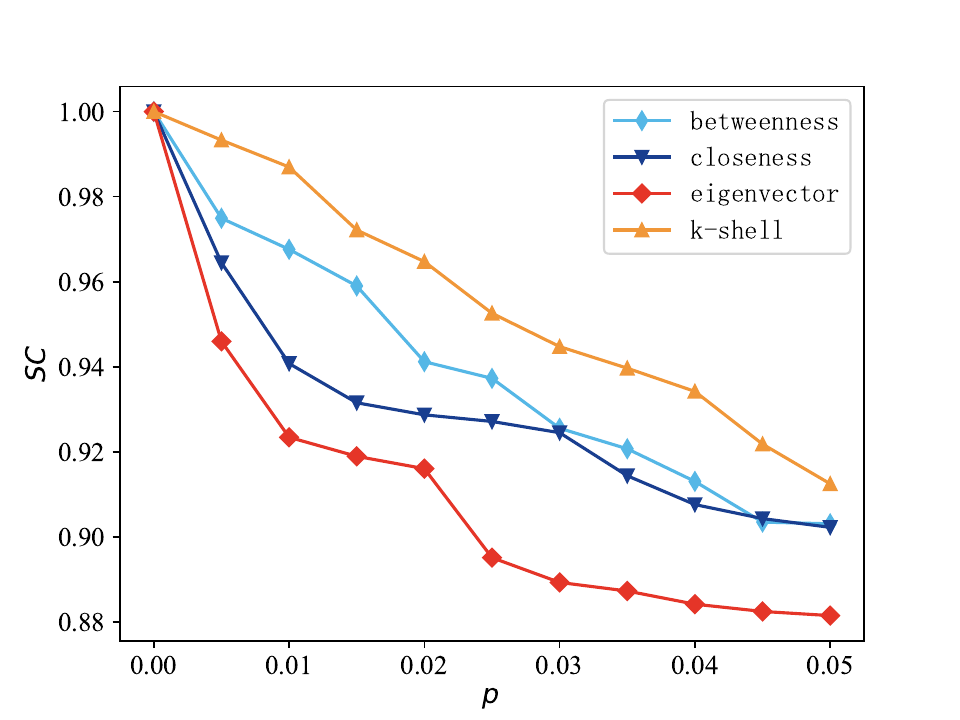}
        \end{minipage}
    }
	\subfloat[BCSPWR10]{
        \begin{minipage}[b]{.3\linewidth}
            \centering
            \includegraphics[scale=0.3]{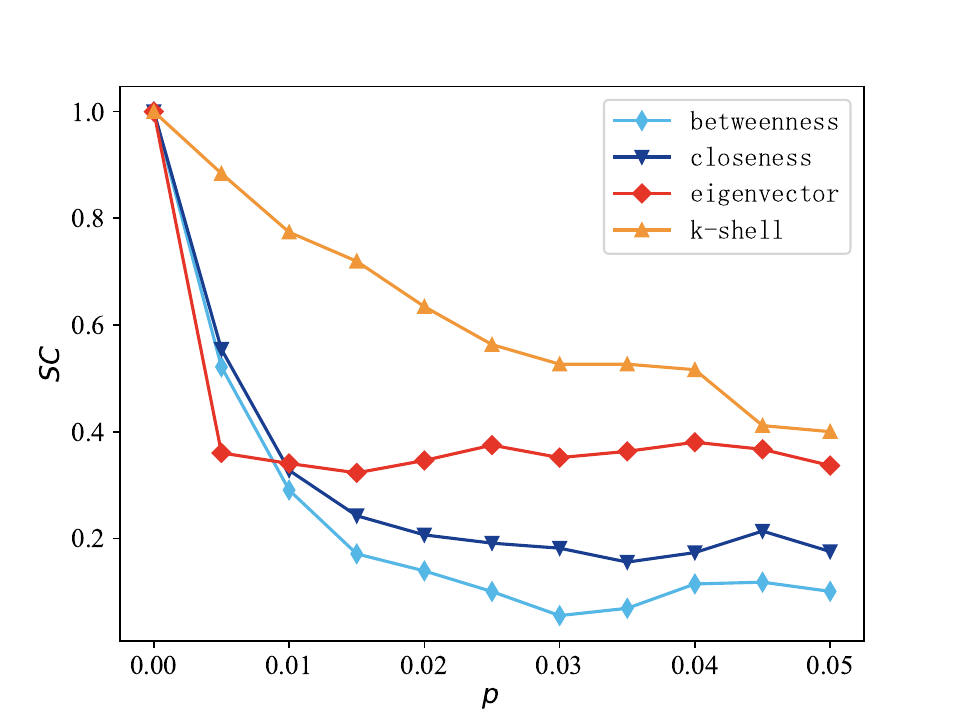}
        \end{minipage}
    }
  % \hspace{2cm}
    \subfloat[USAir10]{
        \begin{minipage}[b]{.3\linewidth}
            \centering
            \includegraphics[scale=0.3]{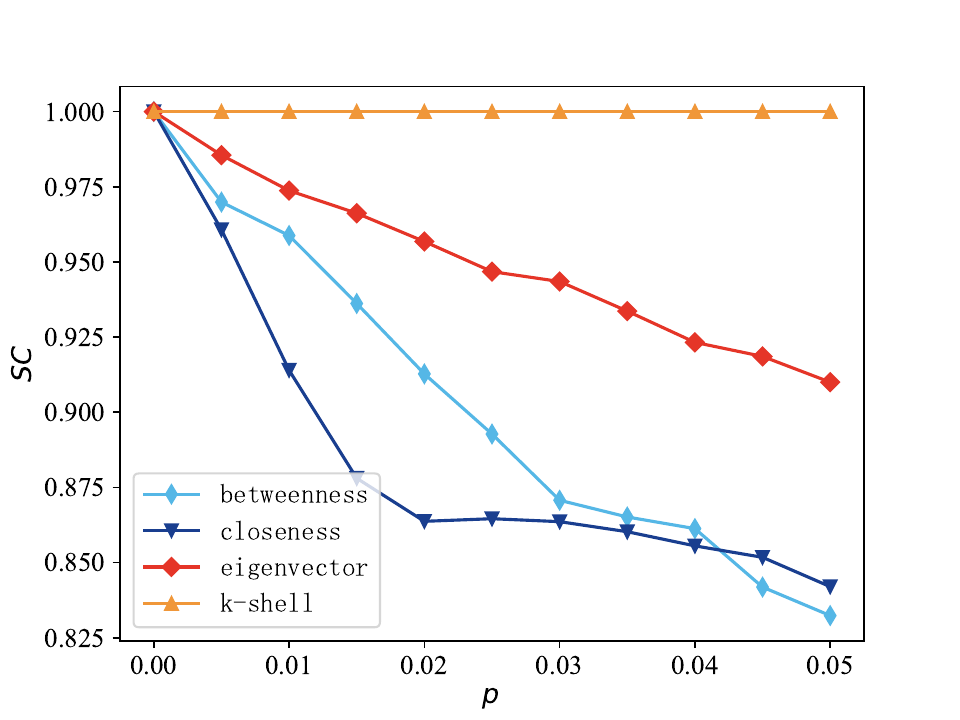}
        \end{minipage}	
    }
	\caption{The Spearman rank correlation coefficient $SC$ between the true centrality measure $C_T$ and the manipulated centrality measure $C_M$ of top-degree nodes, resulting from rewiring edge pairs using the GA method, is analyzed. The rewiring frequencies range from 0.5\% to 5\% of the total number of edges in the network.}
	\label{fig:2222} 
\end{figure*}

Through our previous experiments, we have validated that the GA method can effectively enhance the degree correlation of networks of different types while simultaneously improving their robustness. An interesting question arises: when we optimize network structure using the GA method, can various centrality measures of the network maintain their robustness?

The impact of using the GA method to rewire networks to enhance network degree correlation while affecting centrality measures is illustrated in Figure \ref{fig:2211}. As the number of rewirings increases, the Spearman correlation coefficient $SC$ for all centrality measures initially experiences a rapid decrease before reaching a relatively stable state. One key observation is that across all three types of networks, the robustness of closeness centrality and eigenvector centrality to changes is superior to that of betweenness centrality and k-shell. Especially for routing networks, the $SC$ of closeness centrality and eigenvector centrality can be maintained above 0.8. However, in power networks and flight networks, as the number of rewiring iterations increases, our centrality measures fail to maintain their robustness. We also observed that in disassortative networks, the variations in closeness centrality and eigenvector centrality were similar, indicating a certain correlation between these two centrality measures in disassortative networks.

In fact, in many cases, nodes ranking at the top are more important. Therefore, for each centrality measure, we only consider the robustness of the top 5\% ranked nodes under different rewiring frequencies. It can be observed that for routing networks and flight networks, all four centrality measures remain relatively stable. At a rewiring frequency of 5\%, the $SC$ of all centrality measures is above 0.73. However, in the power network, at a rewiring frequency of 5\%, the $SC$ of all centrality measures is below 0.6. This indicates that the centrality of top-ranked nodes in disassortative networks is more robust. This is because in disassortative networks, the centrality measures of top nodes often exhibit significant numerical differences, making it difficult for nodes with lower centrality measures to surpass others through rewiring. We also found that in the flight network, the k-shell centrality remained robust during the rewiring process. This is because in the flight network, there are numerous connections between high-degree nodes, which typically have higher k-shell. Therefore, rewiring hardly changes their k-shell. Additionally, in the power network, the k-shell also exhibits greater stability compared to other centrality measures.

In the power network, none of the centrality measures can maintain robustness. One possible reason is that in the power network, the degrees of different nodes are relatively close, and the centrality measures of different nodes do not differ significantly in numerical value. When using the GA method for rewiring, it is easier to enhance the centrality of nodes with lower centrality measures, effectively improving their ranking in the respective centrality measure.

% the top-degree nodes are typically already connected to other top nodes, and rewiring only affects their connections to other nodes to a lesser extent, thus having a minor impact on their centrality metrics that consider global information. 

% We observed that, for the neutral flight network OpenFlights, only the k-shell is relatively unstable. This is because many top-degree nodes in OpenFlights are not initially connected, and rewiring through the GA method increases their connections, resulting in significant improvements in the k-shell centrality of some nodes.

\section{Conclusion}\label{thi}
In this work, we addressed the problem of maximizing network degree correlation through a limited number of rewirings while preserving the network degree distribution. We employed the widely used assortativity coefficient to quantify network degree correlation and demonstrated its equivalence to the $s-$metric under degree-preserving conditions. We analyzed the factors that influence changes in the assortativity coefficient under degree-preserving conditions. Based on our assumptions, we formulate the problem of maximizing the assortativity coefficient and verify its monotonic submodularity. Introducing the GA method, we showed through various experiments that it efficiently approximates the optimal solution and outperforms several heuristic methods in enhancing network degree correlation. Additionally, we proposed three heuristic rewiring methods, EDA, TA and PEA, aimed at enhancing network degree correlation. Experimental results revealed that TA is suitable for power networks, while PEA performs well in AS routing networks, and both heuristic methods outperform other baseline methods in flight networks.

We also investigated the impact of our rewiring strategies on network spectral robustness, thus expanding the application scenarios of our approaches. Experimental results demonstrated that our GA strategy effectively enhances both network degree correlation and spectral robustness across all three network types. Particularly, the proposed TA exhibited excellent performance in power networks, even surpassing the GA strategy. We analyzed whether several centrality measures can maintain robustness when the GA method rewires networks. We found that, for disassortative real networks, closeness centrality and eigenvector centrality are typically robust, whereas none of the centrality measures are robust for neutral power grids. When focusing on the top-ranked nodes, we observed that all centrality measures remain robust in disassortative networks.
% Our PEA strategy showed promising results in routing and airline networks. Moreover, we observed that network spectral robustness metrics are not only correlated with degree correlation but also influenced by different rewiring strategies and network types, suggesting potential avenues for further research in subsequent work.

In future work, we also plan to extend the application of our rewiring strategies to fields such as information propagation, exploring whether different rewiring strategies have varying impacts on network dynamic processes. Additionally, regarding altering network degree correlation, we intend to investigate different approaches for modifying network topology, such as adding or deleting edges, to understand how they affect network degree correlation.

% We analyzed the factors that influence changes in the assortativity coefficient under degree-preserving conditions. Based on our assumptions, we formulate the problem of maximizing the assortativity coefficient and verify its monotonic submodularity. From this we proposed a greedy rewiring strategy. Our experimental results provided strong evidence supporting the validity of our assumptions and demonstrated that our algorithm achieves results very close to the optimal solution. Our algorithm exhibited the best performance in both synthetic and real networks.

%The experimental results showed that GRS had the best performance on both synthetic and real networks, and it often compute an optimal solution. 
% \newpage
%\section{References Section}
% \bibliographystyle{IEEEtran}
%\balance
\bibliography{ref}
\bibliographystyle{IEEEtran}

\end{document}